\documentclass[10pt, twocolumn,superscriptaddress,prx]{revtex4-2}
\usepackage[utf8]{inputenc}
\usepackage[T1]{fontenc}
\usepackage{amsmath}
\usepackage{physics}
\usepackage{amsthm}
\usepackage{amsfonts}
\usepackage{amssymb}
\usepackage{bbm}
\usepackage{xcolor}
\usepackage{booktabs}   
\usepackage{tabularx} 
\usepackage{graphicx}
\usepackage[normalem]{ulem}
\usepackage{cancel}

\usepackage[colorlinks=true,linkcolor=teal,citecolor=purple]{hyperref}
\usepackage{cleveref}

\newtheorem{theorem}{Theorem}

\newtheorem{lemma}{Lemma}

\definecolor{plbcol}{rgb}{0.8,0,0}
\definecolor{mplcol}{rgb}{0.,0.8,0.}
\definecolor{gbcol}{rgb}{0,0,1}
\definecolor{ghcol}{rgb}{0.5,0.8,0.0}
\definecolor{nbcol}{rgb}{0.2,0.4,0.4}
\definecolor{jlpcol}{rgb}{0.7,0.2,0.2}

\begin{document}

\title{Thermodynamic Networks: \\ Harnessing Non-Equilibrium Steady States for Computation}
\author{Patryk Lipka-Bartosik}
\affiliation{Center for Theoretical Physics, Polish Academy of Sciences, Warsaw, Poland}
\affiliation{Institute of Theoretical Physics, Jagiellonian University, 30-348 Krak\'ow, Poland}
\affiliation{Department of Applied Physics, University of Geneva, 1211 Geneva, Switzerland}

\author{Gianmichele Blasi}
 \affiliation{Instituto de Física Interdisciplinar y Sistemas Complejos IFISC (CSIC-UIB), E-07122 Palma de Mallorca, Spain}
\affiliation{Department of Applied Physics, University of Geneva, 1211 Geneva, Switzerland}

\author{Javier Lalueza Puértolas}
\affiliation{F\'isica Te\`orica: Informaci\'o i Fen\`omens Qu\`antics, Department de F\'isica, Universitat Aut\`onoma de Barcelona, 08193 Bellaterra (Barcelona), Spain}

\author{Géraldine Haack}
\affiliation{Department of Applied Physics, University of Geneva, 1211 Geneva, Switzerland}

\author{Mart\'i Perarnau-Llobet}
\affiliation{F\'isica Te\`orica: Informaci\'o i Fen\`omens Qu\`antics, Department de F\'isica, Universitat Aut\`onoma de Barcelona, 08193 Bellaterra (Barcelona), Spain}
 \affiliation{Department of Applied Physics, University of Geneva, 1211 Geneva, Switzerland}
 
 \author{Nicolas Brunner}
  \affiliation{Department of Applied Physics, University of Geneva, 1211 Geneva, Switzerland}

\date{\today}

\begin{abstract}
We introduce thermodynamic networks, a general framework for autonomous, physics-based computation using non-equilibrium steady states. These networks are modeled as a collection of finite-size reservoirs that exchange conserved quantities--such as electric charge or molecular number--while relaxing to a non-equilibrium steady state, which encodes the solution of a computational problem. We identify Negative Differential Conductance (NDC) as the critical physical property governing the computational expressivity of the thermodynamic network. While networks lacking NDC are restricted to computing monotonic functions, the presence of NDC enables universal function approximation. For the training of the network, we use  protocols that take advantage of the natural tendency of the system to equilibrate. We illustrate the versatility of our approach via two different platforms: quantum dot networks and enzymatic reaction networks. Both systems can be engineered to have NDC, enabling high performance in standard benchmarks, including sine function approximation and MNIST digit classification. Overall, our work establishes a rigorous link between non-equilibrium steady states and computational expressivity.
\end{abstract}

\maketitle

\section{Introduction}

Physics-based computing exploits the inherent complexity of physical systems for computation. Indeed,
when a physical system relaxes to its steady state, it is effectively solving an equation. An electric circuit settles to the current distribution satisfying Kirchhoff's laws. A mixture of chemicals evolves to the composition that minimizes free energy. Protein concentrations in a living cell stabilize when synthesis and degradation balance. In each of these cases, the system evolves toward a state that satisfies a set of equations specified by its underlying physical laws. If we can design which equations the system solves, e.g. by tuning its properties and interactions, then its relaxation can be exploited as a \emph{computation}. These ideas have received renewed interest, in particular motivated by advances in machine learning and the search for energy-efficient hardware~\cite{kaspar2021rise,wright2022deep,markovic2020physics}.

A wide variety of models for physics-based computing have been developed, each exploiting a different regime of physical dynamics. For example, reservoir computing takes advantage of the nonlinear dynamics of classical or quantum systems for tasks such as time-series forecasting~\cite{Tanaka_2019}, 
while the ground states of many-body systems can be designed to encode solutions to complex optimization problems~\cite{mohseni2022ising}. More recently, the paradigm of thermodynamic computing has emerged~\cite{conte2019thermodynamic}, in which Gibbs states and stochastic dynamics~\cite{seifert2012stochastic} are used to tackle problems in linear algebra~\cite{aifer2024thermodynamic,Moroder2026}, convex optimisation~\cite{lipka2024thermodynamicquadratic} and machine learning~\cite{coles2023thermodynamic,aifer2024thermodynamicbayesianinference,whitelam2026generative,whitelam2026nonlinear,rolandi2026energytimeaccuracytradeoffsthermodynamiccomputing}. Another approach proposed a physics-based analog implementation of neural networks using open quantum systems~\cite{lipka2024thermodynamic}.

In this work, we introduce \emph{thermodynamic networks}, a general framework for physics-based computing based upon non-equilibrium steady states (NESS). As illustrated in Fig.~\ref{fig:main}, we model such a network as a collection of mesoscopic reservoirs (the nodes) characterized by thermodynamic potentials, e.g., temperature or chemical potential. These reservoirs are coupled by transport channels (the edges) which enable the exchange of conserved quantities such as energy or particles. Driven by potential differences, the network evolves to a NESS. From the latter one can extract the result of a computational problem, which had been initially encoded in the configuration of the network. The entire process is thus autonomous, governed only by dissipation without the need for a clock or an external processor.

 \begin{figure*}[!t]
\centering
\includegraphics[width=\textwidth]{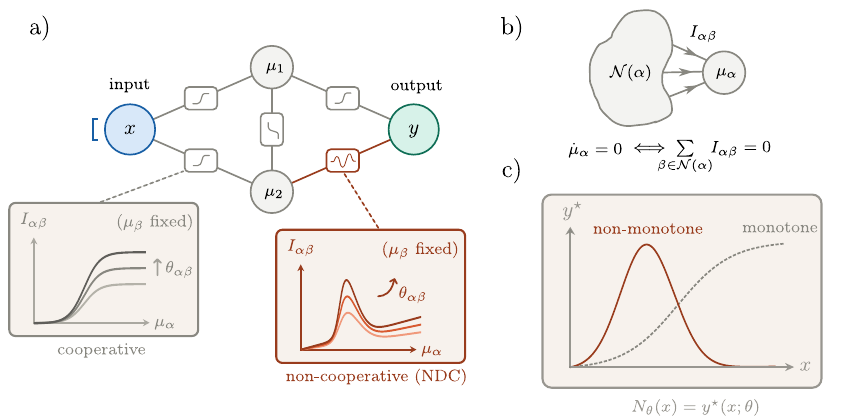}
\caption{\textbf{Thermodynamic networks compute by equilibrating.}  \textbf{(a)}~A thermodynamic network is a graph of reservoirs (nodes) coupled by transport channels (edges) that collectively redistribute a conserved charge. The conjugate potential $\mu$ serves as both the dynamical state variable at every node and the language of input and output. The input reservoir $x$ is clamped to a prescribed potential (bracket), hidden nodes~$\mu_1$ and~$\mu_2$ evolve freely, and the output~$y$ is read out at steady state. Each transport channel (rectangular box) carries a current $I_{\alpha\beta}(\mu_\alpha, \mu_\beta)$ governed by tuneable parameters~$\theta_{\alpha\beta}$. Insets show the current--potential characteristic at fixed $\mu_{\beta}$ for a family of parameter values. Cooperative channels (left) yield a monotone response, whereas a channel with negative differential conductance (right) exhibits a rise and fall profile. \textbf{(b)} Current conservation at a single node. At steady state, the potential~$\mu_\alpha$ adjusts until the currents from all neighbours~$\beta \in \mathcal{N}(\alpha)$ balance, so that~$\dot{\mu}_\alpha = 0$. \textbf{(c)} Input--output map $N_\theta(x) = y^\star(x;\theta)$ in the steady state. Cooperative networks are constrained by the Kamke--M\"uller theorem to monotone functions (dashed grey). Single NDC channel breaks this constraint, unlocking non-monotone maps (solid red) and qualitatively richer computational expressivity.}
\label{fig:main}
\end{figure*}

Our central result is that the expressivity of thermodynamic networks is governed by a single physical property: the differential conductance of the transport channels, which quantifies how currents respond to changes in the neighboring potentials. This quantity divides thermodynamic networks into three regimes of increasing expressivity. When all differential conductances are constant, currents are proportional to potential differences, and the network can only compute linear functions. When differential conductances are nonlinear but remain positive, the network is cooperative: it can represent nonlinear functions, but only monotonic ones. A qualitative transition occurs when at least one channel exhibits negative differential conductance (NDC), where increasing the potential difference decreases the current. This breaks monotonicity and makes the network a universal function approximator, capable of representing any continuous function to arbitrary precision. 

The structure of thermodynamic networks naturally suggests the use of training methods that depart from conventional machine learning. Specifically, one can exploit the steady-state condition to replace gradient computations with solving a linear system at a fixed point. This approach is closely related to Deep Equilibrium Models~\cite{bai2019deepequilibriummodels}, and allows us to transfer the toolkit of fixed-point neural networks to thermodynamic networks with only minor modifications.

To demonstrate the framework's versatility, we illustrate it via two distinct platforms. First, we consider \emph{networks of quantum dots}, where mesoscopic electronic reservoirs (the nodes) are coupled with transport channels (the edges), corresponding to quantum dots capacitively coupled to the leads. Here NDC emerges through a potential-dependent electrostatic effect, which shifts the energy level of the dot. Second, we consider \emph{enzymatic reaction networks}, where molecular pools (the nodes) exchange chemical species with the help of enzyme-catalyzed reactions (the edges). In this case, NDC arises through substrate inhibition, where excess substrate blocks the enzyme's active site. Despite their different microscopic physics, both platforms exhibit a similar activation in computational power: the presence of NDC breaks monotonicity and enables universal approximation. We verify this prediction on standard benchmarks, including XOR, sine regression, and classification on the MNIST dataset. 

\section{Thermodynamic networks}
\label{sec:model}
Under suitable conditions, such as dissipation to coupled reservoirs, physical systems driven away from equilibrium can relax to steady states in which competing currents balance out. This relaxation can be engineered so that the steady state encodes the solution to a computational problem. We now formalize this idea through the framework of \emph{thermodynamic networks}.

\subsection{Network structure and dynamics}

A \emph{thermodynamic network} is a graph $G = (V, E)$ whose nodes $V$ represent mesoscopic reservoirs and whose edges $E$ represent transport channels. Each node $\alpha \in V$ stores a conserved quantity $q_\alpha$ (e.g. particle number) and is characterized by a conjugate potential $\mu_\alpha$ (e.g. chemical potential). These variables are related via
\begin{align}
\label{eq:capacitance}
\mathrm{d} q_{\alpha} = C_{\alpha} \, \mathrm{d} \mu_{\alpha},
\end{align}
where $C_{\alpha} > 0$ is the generalized capacity of the reservoir. 

Each edge $(\alpha, \beta) \in E$ carries a current $I_{\alpha\beta}(\mu_\alpha, \mu_\beta)$ that depends on the potentials of the two connected reservoirs. The functional form of $I_{\alpha\beta}$ captures the microscopic physics of the transport channel, with the dynamics of the network being independent of its specific origin. 

It is convenient to collect the contributions of all edges meeting at node $\alpha$ into
\begin{align}
\label{eq:node_current}
I_\alpha(\mu) := \sum_{\beta\in\mathcal{N}(\alpha)} I_{\alpha\beta}(\mu_\alpha,\mu_\beta),
\end{align}
where $\mu$ is the vector of all potentials and $\mathcal{N}(\alpha)$ denotes the neighbors of $\alpha$. Definitions of charge and current imply that $\dot q_\alpha = -I_\alpha$, where negative sign means that outgoing current depletes the reservoir. Combining this with Eq.~\eqref{eq:capacitance} yields the network dynamics,
\begin{align}
\label{eq:dynamics}
C_\alpha \dot\mu_\alpha = -I_\alpha(\mu).
\end{align}
We collect all node potentials into a vector $\mu$ and write $F_\alpha(\mu) := 
-I_\alpha(\mu)/C_\alpha$, so that the network's dynamics takes the compact form $\dot\mu = F(\mu)$.

In the following, when we need to emphasize the role of tunable parameters $\theta$, we will write explicitly $F(\mu; \theta)$. Some nodes may further be clamped, meaning their potentials are held fixed, so that $\mu_{\alpha}(t) = x_{\alpha}$ for all $t$.

The currents are further assumed to respond instantaneously to changes in the potentials, which we refer to as the \emph{adiabatic regime} of a thermodynamic network. This assumption is valid when the internal relaxation time of each edge is much shorter than the node charging time. This separation of timescales is well justified in many relevant scenarios, as discussed in details in Appendix~\ref{app:microscopic}.

\subsection{Computing via equilibration}
In the following, we assume that as the network relaxes, it approaches a NESS. In this state, the time derivatives of the potentials vanish, so that
\begin{align}
    \label{eq:fixed_point}
    F(\mu^\star; \theta) = 0.
\end{align}
The steady state potential $\mu^{\star} = \mu^{\star}(x; \theta)$ that solves Eq.~\eqref{eq:fixed_point} depends on the network parameters $\theta$, and the clamped (fixed) potentials $x$ of input nodes. 

Let us now explain how the network performs computation. Consider the partition of nodes as $V = V_{\text{in}} \cup V_{\text{hid}} \cup V_{\text{out}}$, where input nodes $V_{\text{in}}$ are clamped to prescribed values, hidden nodes $V_{\text{hid}}$ are free to evolve, and output nodes $V_{\text{out}}$ are read out after the steady state is reached. The network then implements a function
\begin{align}
\label{eq:network_function}
N_\theta (x)  :=  \mu^\star_{\text{out}}(x;\, \theta),
\end{align}
which maps inputs to steady-state outputs via physical relaxation. Training the network means choosing the parameters $\theta$ so that $N_{\theta}$ approximates the desired target function. 

\subsection{Differential conductance}
Which functions can a  thermodynamic network compute? To answer this, consider the Jacobian matrix of network dynamics~\eqref{eq:dynamics}. Its elements,
\begin{align}
J_{\alpha\gamma}(\mu) = \frac{\partial F_\alpha}{\partial \mu_\gamma},   
\end{align}
describe how a change in the potential at node $\gamma$ affects the rate of change at node $\alpha$. Inserting $F_{\alpha} = -I_{\alpha}/C_{\alpha}$ in the previous equation and differentiating, we find
\begin{align}
\label{eq:jacobian}
J_{\alpha\gamma} = \frac{1}{C_\alpha}\, G_{\alpha\gamma}, 
\quad \text{where} \quad
G_{\alpha\gamma} := -\frac{\partial I_\alpha}{\partial \mu_\gamma}.
\end{align}
The matrix $G$ is the network's \emph{conductance matrix} which measures how the currents at one node respond to changes of potential at another. Furthermore, the matrix element $G_{\alpha\gamma}$ is nonzero only when $\gamma=\alpha$ or $\gamma\in\mathcal{N}(\alpha)$, so $G$ inherits the sparsity pattern of the network's graph. 

\subsection{Cooperative dynamics}
\label{sec:cooperativity}
The sign structure of the conductance matrix imposes a fundamental constraint on the class of functions that can be implemented by a thermodynamic network. 

A dynamical system $\dot{\mu} = F(\mu)$ is \emph{cooperative} if all off-diagonal elements of its Jacobian are non-negative~\cite{smith1995monotone}, namely $J_{\alpha\gamma}\ge 0$ 
for all $\alpha\neq\gamma$ and all $\mu$. By Eq.~\eqref{eq:jacobian}, this translates into a condition on the conductance matrix,
\begin{align}
\label{eq:cooperative}
G_{\alpha\gamma}(\mu) \ge 0 
\qquad \text{for all }\ \gamma\neq\alpha,\ \text{and all }\mu.
\end{align}
This is the natural multi-terminal generalisation of the fact that current flows from high to low potential. Cooperativity has important consequences, as captured by the following result from the theory of dynamical systems.

\begin{theorem}[Kamke-Müller~\cite{smith1995monotone}]
\label{thm:kamke}
If the network's dynamics $\dot{\mu} = F(\mu)$ is cooperative and $\mu(0) \leq \mu'(0)$ componentwise, then $\mu(t) \leq \mu'(t)$ for all $t > 0$.
\end{theorem}
The above theorem states that cooperative vector fields preserve the partial ordering of potentials $\mu(t)$ globally and for all times $t$. For thermodynamic networks this means that if two inputs satisfy $x \leq x'$ componentwise, then the corresponding steady states satisfy $\mu^\star(x) \leq \mu^\star(x')$, and the network function $N_\theta(x)$ is necessarily monotone in every input $x$.

The Kamke-Müller theorem implies that cooperative thermodynamic networks can only compute monotone functions---not by a choice of parameters, but by the structure of the dynamics itself. Therefore, no amount of training can produce a non-monotone input-output map if the network dynamics is cooperative. For certain problems this is not a limitation but a feature. For instance, in healthcare triage or credit scoring, a model must not assign lower priority to a patient with a more severe illness, or lower creditworthiness to an applicant with a higher
income~\cite{liu2020certified,runje2023constrained,chen2022monotonic}. Conventional neural networks offer no such guarantee---enforcing monotonicity typically requires either restricting the architecture (e.g. constraining all weights to be non-negative, thus making the model difficult to train) or expensive post-hoc verification~\cite{liu2020certified}. In contrast, in cooperative thermodynamic networks, monotonicity emerges from physics, in particular through transport theory and conservation laws.

An example of a cooperative thermodynamic network is the network operating in the \emph{linear transport} regime, where $I_{\alpha\gamma} = G_{\alpha\gamma}(\mu_\alpha - \mu_\gamma)$ with constant $G_{\alpha\gamma} > 0$. It can be easily shown that for such networks, the network function $N_\theta$ is also a linear map.

For general computation, however, monotonicity is a severe constraint. The exclusive-or function (XOR) is a simple example of a non-monotone function. Since for inputs $x = (1,1)$ and $x' = (0,1)$ we have $x \geq x'$ componentwise, monotonicity of $N_{\theta}$ means $N_\theta(x) \geq N_\theta(x')$, whereas XOR requires the opposite. No cooperative network can therefore compute XOR. In the next section, we discuss in detail how to overcome this limitation.

\subsection{Negative differential conductance}
\label{sec:ndc}
A non-monotone map $N_{\theta}(x)$ requires that increasing the input \emph{decreases} the output, which is precisely the behavior forbidden by Theorem~\ref{thm:kamke}. To break cooperativity,  one then needs a mechanism to violate Eq.~\eqref{eq:cooperative}. This means that at some operating point $\mu$ and some neighbour $\gamma\in\mathcal{N}(\alpha)$, we must have
\begin{align}
\label{eq:ndc_def}
G_{\alpha\gamma}(\mu) < 0.
\end{align}
This is the phenomenon of negative differential conductance (NDC)~\cite{esaki1958new,esposito2009nonequilibrium,baiesi2009fluctuations}, a generic property of nonlinear transport in which raising the potential difference across a channel decreases the current through it. It arises in a wide range of physical systems, such as semiconductor tunnel diodes~\cite{esaki1958new}, quantum dots~\cite{van2002electron} or enzymatic reactions~\cite{reed2010biological}. Although microscopic origins of NDC differ among systems, its consequences at the level of the network are universal: a single edge with $G_{\alpha \gamma} < 0$ is enough to violate the conditions of Theorem~\ref{thm:kamke} and allows the map $N_{\theta}(x)$ to reverse the ordering of inputs. Geometrically, NDC allows for \emph{folding} the input-output map (see Fig.~\ref{fig:main}c), providing the mechanism that breaks monotonicity and unlocks expressive computation. 

A striking computational feature of NDC is that the folds it enables are \emph{local}. The sign of $G_{\alpha \gamma}(\mu)$ depends on the operating point $\mu$, so the same edge can fold the map $N_{\theta}(x)$ in one region of input space, and act cooperatively in another. This is different from conventional neural network where nonlinearities are fixed by the architecture. In thermodynamic networks nonlinearities are chosen by the steady state and the input selects which edges to fold.
\subsection{Universal approximation}
\label{sec:universal}
We have seen that NDC allows a thermodynamic network to compute non-monotone functions. It is then natural to ask which functions can such a network compute. The answer turns out to be surprising: When NDC is available, a thermodynamic network can approximate any continuous function to arbitrary accuracy.

In Theorem~\ref{thm:universal} of Appendix~\ref{app:universality}, we prove that for any continous target function $f$ defined on a compact domain and any accuracy $\delta > 0$, there exists a thermodynamic network whose steady-state map $N_{\theta}(x)$ satifsfies $\sup_{x} |N_{\theta}(x) -f(x)| < \delta$. Universality of thermodynamic networks is therefore not tied to any particular platform and follows whenever the network has access to NDC, irrespectively of how this property is physically realized.   

\textit{Proof sketch.} Consider a three-layer network composed of input, hidden and output reservoirs such that the hidden nodes are only weakly coupled to the output reservoir. In this regime, the layers decouple to the leading order, so that each hidden node settles to a steady state determined only by the input nodes. These hidden nodes in turn drive the output with a small residual current. A direct calculation then shows that the steady-state potential of the output node has a form similar to the output of a single-hidden-layer neural network. In this analogy, the currents $I_{\alpha}$ play the role of effective activation functions at nodes $\alpha$ and the differential conductances $G_{\alpha \beta}$ set the effective weights of the network. A classical approximation theorem~\cite{cybenko_approximation_1989,hornik1991approximation} then guarantees that such a network can approximate any continuous function, provided its weights can take any real value. This translates into the requirement that differential conductances can take any (both positive and negative) value. 


The above sketch shows precisely how NDC influences the expressivity of thermodynamic networks. Notably, positive differential conductances can only produce positive effective weights, so a cooperative network (see Sec.~\ref{sec:cooperativity}) can realise only non-negative effective weigths and hence only monotone functions. The presence of NDC lifts this restriction and enables negative effective weights that are required for universal approximation. The mechanism that folds the input--output map (Sec.~\ref{sec:ndc}) is therefore the same mechanism that makes thermodynamic networks universal.

\section{Training thermodynamic networks}
\label{sec:training}
The map $N_\theta(x)$ implemented by a thermodynamic network depends on a set of tunable physical parameters $\theta$ that specify its action. Training the network means adjusting these parameters so that $N_\theta(x)$ accurately approximates the desired target function.

\subsection{Gradient descent}
Given a dataset $\mathcal{D} = \{(x^k, y^k)\}$ of $K$ data points, training proceeds by minimizing the empirical loss defined as
\begin{align}\label{eq:loss}
L(\theta) = \frac{1}{K}\sum_{k=1}^K \mathcal{L} \big(N_\theta(x^{k}),\, y^{k}\big),
\end{align}
where $\mathcal{L}(\cdot, \cdot)$ measures the mismatch between network's output and the target, e.g. $\mathcal{L}(a, b) = \|a - b\|^2$. The typical strategy for minimizing loss is gradient descent, which iteratively updates the parameters as
\begin{align}
    \theta \leftarrow \theta - \eta\, \nabla_\theta L,
\end{align}
with $\eta > 0$ being the learning rate. The problem of training therefore reduces to computing the gradients $\nabla_{\theta} L$.

The simplest approach to computing gradients is via finite differences. In the context of thermodynamic networks this means perturbing each parameter $\theta_i$ by a small $\varepsilon>0$ and relaxing the network to a new steady state. This yields the $i$-th component of the gradient, namely
\begin{align}
    \frac{\partial L}{\partial \theta_i}
\approx \frac{L(\theta+\varepsilon e_i) - L(\theta)}{\varepsilon},
\end{align}
where $e_i$ is the $i$-th standard basis vector. This procedure is universal, but also very expensive. For $m$ parameters and $K$ training examples, it requires $\mathcal{O}(m K)$ independent relaxations of the whole network. For most realistic problems this quickly becomes infeasible. 

In conventional neural networks, the same scaling problem is solved by \emph{backpropagation}. There, the error signal is propagated backward through well-defined layers, and a single backward pass gives gradients with respect to all parameters~\cite{rumelhart1986}. The key assumption of the method is that the network can be seen as a sequence of operations that can be unrolled and reversed. Thermodynamic networks have no such structure, since their action is defined implicitly through the steady-state equation~\eqref{eq:fixed_point}. Thus, the very feature that makes thermodynamic networks physically plausible, i.e. computation by equilibration, also obstructs standard machine learning approaches for training.

\subsection{Implicit differentiation}
\label{sec:implicit_diff}
 The steady-state equation suggests a more efficient route to compute gradients. Because $\mu^{\star}$ is defined implicitly as the solution to Eq.~\eqref{eq:fixed_point}, its derivatives with respect to $\theta$ can be computed directly using the implicit function theorem. This is the approach behind \emph{Deep Equilibrium Models} (DEQs)~\cite{bai2019deepequilibriummodels}, a class of machine learning models defined by fixed-point equations rather than by a sequence of layers. In Appendix~\ref{app:training} we adapt this training method to thermodynamic networks and sketch the basic idea below.

The starting point is to differentiate the steady-state condition from Eq.~\eqref{eq:fixed_point} with respect to parameter $\theta_k$. This yields a linear system that determines how the steady state shifts when $\theta_k$ is varied. However, solving it once for each parameter independently would not give us any advantage over the finite difference method. The crucial observation is that we do not need the full shift of the steady state. Instead, for training we only need its projection onto the loss gradient at the output. This means that all $m$ parameter gradients  can be obtained from
\begin{align}
    \label{eq:adjoint_system}
    (J^{\star})^{\top} \lambda = P^{\top} g,
\end{align}
where $g :=
\left.
\nabla_z \mathcal{L}(z,y)
\right|_{z=N_\theta(x)} $ is the loss gradient and $P$ is a projector on output nodes. Once $\lambda$ is found, every parameter gradient can be determined through
\begin{align}
\label{eq:adjoint_formula}
\frac{\partial L}{\partial \theta_k} 
= -\lambda^\top \frac{\partial F}{\partial \theta_k}\bigg|_{\mu^\star},
\end{align}
regardless of the number of parameters.

The vector $\lambda$ has an interesting physical interpretation. It measures how the loss responds to a small perturbation applied at each node and propagated through the network. It plays the role of the backpropagated error signal of conventional neural networks, but distributed not through a sequence of layers, but instead through the network's own response structure.

Implicit differentiation offers a practical computational advantage. The adjoint 
system from Eq.~\eqref{eq:adjoint_system} can be solved iteratively, for instance using conjugate-gradient methods. Furthermore, the entire procedure requires only \emph{a single} equilibration per training example, irrespective of how many parameters need to be trained. For $m$ parameters and $K$ training examples, implicit differentiation therefore costs $\mathcal{O}(K)$ relaxations of the network.

Implicit differentiation is our method of choice for numerical simulations presented in the remaining part of this work. Its main limitation is that it requires explicit access to the Jacobian matrix, which is straightforward in simulation but may be inaccessible in a physical device. This motivates the complementary approach via equilibrium propagation~\cite{scellier_equilibrium_2017}, which we discuss in Appendix \ref{app:training}. 

\section{Physical realizations of thermodynamic networks}
Having developed the abstract framework of thermodynamic networks, we now apply it to two distinct platforms: quantum dots and enzymatic reactions. A reference implementation, including the network solver, the training routines, and scripts that reproduce the results below, is available at~\cite{thermonet-code}.

\subsection{Thermodynamic networks of quantum dots}
\label{sec:quantum_dots}
As our first example we consider a network of quantum dots connecting electronic or fermionic reservoirs. Transport through such mesoscopic conductors is commonly described within the Landauer–Büttiker framework~\cite{buttiker1986four,Blasi_PRR2024,imry1999conductance}, which captures a wide range of nanoscale transport phenomena~\cite{nazarov2009quantum}. Within this setting, quantum-dot networks provide a paradigmatic yet experimentally accessible realization of thermodynamic networks.

\subsubsection{Physical setup}
The network consists of electronic reservoirs (the nodes) connected by single-level quantum dots (the edges), as shown Fig.~\ref{fig:qd_realisation}. A reservoir $\alpha$ is characterized by a chemical potential $\mu_\alpha$ and temperature $T_\alpha$. Each quantum dot sitting on the edge connecting two reservoirs $(\alpha, \beta)$ has a single discrete energy level $\epsilon_{\alpha\beta}$ through which electrons tunnel one at a time.

\begin{figure}[t]
  \centering
  \includegraphics[width=\columnwidth]{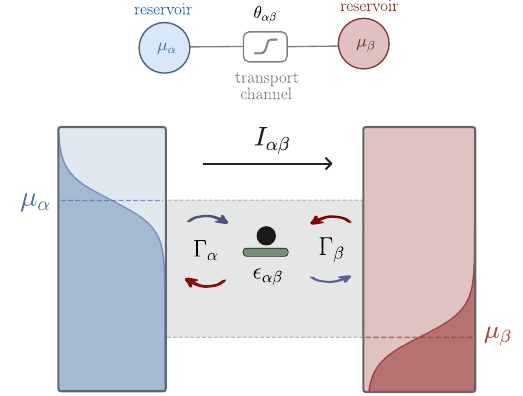}
  \caption{\textbf{Quantum-dot realisation of a thermodynamic network edge.} Physical implementation of an edge $(\alpha, \beta)$ using a quantum dot. Each reservoir is an electronic lead characterised by a chemical potential~$\mu_\ell$ and temperature~$T_\ell$. The shaded regions represent their Fermi--Dirac occupation. The dot hosts a single energy level~$\epsilon_{\alpha \beta}$ through which electrons tunnel sequentially, with rates set by the tunnel couplings $\Gamma_\alpha$ and $\Gamma_\beta$. The dot energy level further shifts with the applied bias via the electrostatic shift effect, see Eq. \eqref{eq:level_energy}. For certain choices of parameters, the energy level can be swept out of the transport window (grey shading), producing negative differential conductance---the mechanism that breaks cooperativity and enables non-monotone computation.}
  \label{fig:qd_realisation}
\end{figure}
A subset of reservoirs serves as inputs, with chemical potentials clamped to prescribed values. The remaining reservoirs exchange electrons through the quantum dots until a NESS is reached. The chemical potentials of designated output reservoirs then encode the result of the computation. The dynamics is governed by Eq.~\eqref{eq:dynamics}, with the current now specified by the microscopic physics of the quantum dot.

\subsubsection{Transport through a quantum dot}
The steady-state current through a quantum dot connecting reservoirs $\alpha$ and $\beta$ can be expressed as the weak-coupling limit of coherent transport through a single resonant level,
\begin{align}
\label{eq:landauer}
I_{\alpha\beta} = \gamma_{\alpha\beta} \big[ f_\alpha(\epsilon_{\alpha\beta}) - f_\beta(\epsilon_{\alpha\beta}) \big],
\end{align}
where $\gamma_{\alpha\beta}$ is the tunnel coupling  and
\begin{align}
f_\ell(E) = \frac{1}{1 + e^{(E - \mu_\ell)/k_B T_\ell}}
\end{align}
is the Fermi--Dirac distribution of reservoir $\ell$. Equation~\eqref{eq:landauer} provides a simple description of transport through a quantum dot: the current is set by the imbalance between the Fermi occupations of the two reservoirs evaluated at the dot level. As shown in Appendix~\ref{app_LB_to_quantumdot}, it follows directly from the Landauer--B\"uttiker formalism in the weak-coupling regime $\gamma \ll k_B T$~\cite{Datta1997,Blanter2000,van2002electron, Blasi_PRR2024}. Importantly, the Landauer--B\"uttiker formalism extends beyond the weak-coupling resonant-level model considered here. The quantum dot provides a simple and experimentally accessible realization, but the same framework applies to a broader class of coherent conductors whose bias-dependent transmission can give rise to negative differential conductance and can therefore serve as alternative implementations of thermodynamic networks.

\begin{figure*}[t]
\centering
\includegraphics[width=0.9\linewidth]{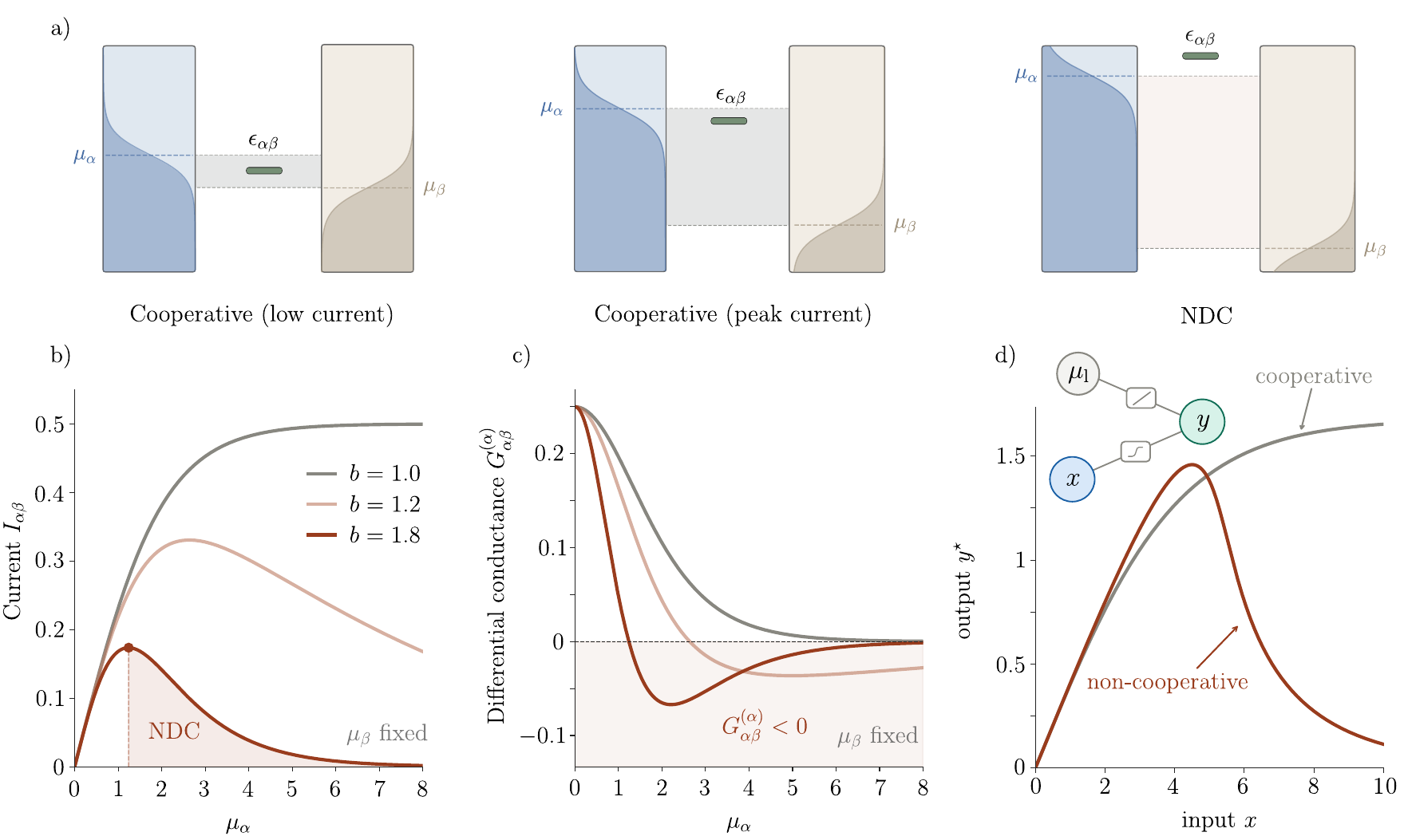}
\caption{\textbf{Negative differential conductance in a quantum dot and its computational consequences}. ~\textbf{(a)} Energy-level diagrams of a quantum dot at three operating points. At low bias $\mu_{\alpha} - \mu_{\beta}$, the dot level $\epsilon_{\alpha \beta}$ lies within the transport window between $\mu_\alpha$ and $\mu_{\beta}$, and current flows freely. As the bias increases, the level is pushed above the source chemical potential, suppressing current despite the growing driving force, which is the hallmark of negative differential conductance (NDC).~\textbf{(b)} Current–voltage characteristics for a quantum dot connected to equal temperature reservoirs $T_{\alpha} = T_{\beta}$ with a varying electrostatic shift parameter $b$. For $b=1$ (no NDC), the current increases monotonically with bias. For $b>1$, a peak develops and the current eventually decreases, leading to NDC.~\textbf{(c)} Source-side differential conductance $G_{\alpha \beta}$ as a function of $\mu_{\alpha}$ (with $\mu_{\beta}$ held fixed). In the cooperative regime $G_{\alpha \beta} > 0$, raising the drain potential suppresses the current, preserving the monotone ordering of the Kamke–Müller theorem. When NDC sets in, (i.e. when $G_{\alpha \beta} < 0$, shaded), raising the drain potential instead enhances the current, violating cooperativity and enabling non-monotone computation.~\textbf{(d)} Steady-state input-output map implemented by a minimal network (one input node $x$ coupled through a quantum dot to one output node $y$ with a linear leak $\mu_l$). The cooperative dot ($b = 1$, blue) produces a monotonically increasing input–output map. The NDC dot ($b = 1.8$, red) produces a map with a fold and enables thermodynamic networks to compute non-monotone functions.} 
\label{fig:iv_curve}
\end{figure*}

\subsubsection{NDC mechanism (electrostatic shift) }
Under certain conditions, the energy level of a quantum dot can be shifted by its electrostatic environment, in particular by the potentials $\mu_{\alpha}$ and $\mu_{\beta}$ of the two reservoirs connected to the dot. Here we model this dependence as a linear \emph{electrostatic shift}, namely
\begin{align}
\label{eq:level_energy}
\epsilon_{\alpha\beta} = a_{\alpha\beta} + b_{\alpha\beta}\,(\mu_\alpha - \mu_\beta),
\end{align}
where $a_{\alpha\beta}$ is the bare energy level and $b_{\alpha\beta}$ quantifies how strongly the dot level tracks the bias established between the two leads (nodes). This is a well-known mechanism in several nanoscale transport platforms, including quantum dots~\cite{kouwenhoven1997electron}, molecular junctions~\cite{evers2020advances}, resonant-tunnelling structures~\cite{sun1998resonant}, and related mesoscopic conductors~\cite{biele2019beyond}. A microscopic derivation of Eq.~\eqref{eq:level_energy} via the Landauer--B\"uttiker framework is given in Appendix~\ref{app_LB_to_quantumdot}.

The electrostatic shift is the mechanism through which quantum dots can break cooperativity. To see this, we compute the off-diagonal conductance matrix element associated with edge $(\alpha, \beta)$. Differentiating Eq.~\eqref{eq:landauer} with respect to $\mu_\beta$ gives (see Appendix~\ref{app:qd_conductances} for details),
\begin{align}
\label{eq:drain_conductance}
G_{\alpha\beta} = \gamma_{\alpha\beta} \Big[ \frac{1 + b_{\alpha \beta}}{T_{\beta}} \tilde{f}_\beta - \frac{b_{\alpha \beta}}{T_{\alpha}} \tilde{f}_\alpha \Big],
\end{align}
where $\tilde{f}_\ell := f_\ell(1 - f_\ell) \geq 0$ is the thermal broadening factor and $\beta_\ell = 1/k_B T_\ell$. The two terms above reflect two competing effects: the term $\tilde{f}_\beta$ captures the usual cooperative response, while the term $\tilde{f}_\alpha$ arises due to the electrostatic shift and allows for breaking cooperativity.

\subsubsection{Breaking cooperativity}
\label{sec:qd_cooperativity}
Cooperativity requires $G_{\alpha \beta} \geq 0$ for all edges and all operating points (see Sec.~\ref{sec:cooperativity}). Using Eq.~\eqref{eq:drain_conductance} it then follows that this is true if and only if
\begin{align}
\label{eq:coop_condition}
\frac{T_\beta \tilde{f}_\alpha}{T_\alpha \tilde{f}_\beta} \leq \frac{1 + b_{\alpha\beta}}{b_{\alpha\beta}}.
\end{align}
The right-hand side specifies the \emph{cooperativity threshold}, namely the maximum asymmetry between source and drain that ensures cooperative behaviour. Notably, this threshold depends on the electrostatic shift parameter $b_{\alpha \beta}$ and the two reservoir temperatures $T_{\alpha}$ and $T_{\beta}$. 

The mechanism responsible for breaking cooperativity is illustrated in Fig.~\ref{fig:iv_curve}. At low bias, the dot level $\epsilon_{\alpha\beta}$ lies between the source and drain potentials, and the current rises with bias as expected, see Fig.~\ref{fig:iv_curve}(a). At larger bias, the electrostatic shift pushes the dot level above the source potential, thereby suppressing transport even though the applied bias continues to increase.  This produces the non-monotonic current--bias relation shown in Fig.~\ref{fig:iv_curve}(b), and the corresponding sign reversal of the differential conductance shown in Fig.~\ref{fig:iv_curve}(c).


\subsubsection{A minimal example: XOR problem}
\label{sec:xor}
We now test the mechanism of NDC on the simplest non-monotone Boolean function: exclusive-or (XOR). For that we consider a network with architecture 2-2-1, i.e. two input reservoirs $(x_1, x_2)$, two hidden reservoirs $(h_1, h_2)$, and one output $(y)$, fully connected by quantum dots with parameters $\{w_{\alpha\beta}, a_{\alpha\beta}, b_{\alpha\beta}\}$. We then train the network using implicit differentiation, as discused in Sec.~\ref{sec:training}. The results are plotted in Fig.~\ref{fig:xor_network}, which clearly demonstrate the difference between cooperative and non-cooperative thermodynamic networks.

Suppose we restrict every dot to the cooperative regime by imposing Eq.~\eqref{eq:coop_condition}. Due to Theorem~\ref{thm:kamke}, the steady-state output $y^{\star}(x_1, x_2)$ is monotone in each input for any choice of parameters. The output is then forced to be monotone in both inputs by Theorem~\ref{thm:kamke}, so the level sets of $y^\star$ are necessarily monotone curves which cannot create a decision boundary separating points $(0,0)$ and $(1,1)$ from $(0,1)$ and $(1,0)$, as required by XOR. Training therefore stalls at a non-zero loss (see Fig.~\ref{fig:xor_network}a). 

\begin{figure}[t]
\centering
\includegraphics[width=\linewidth]{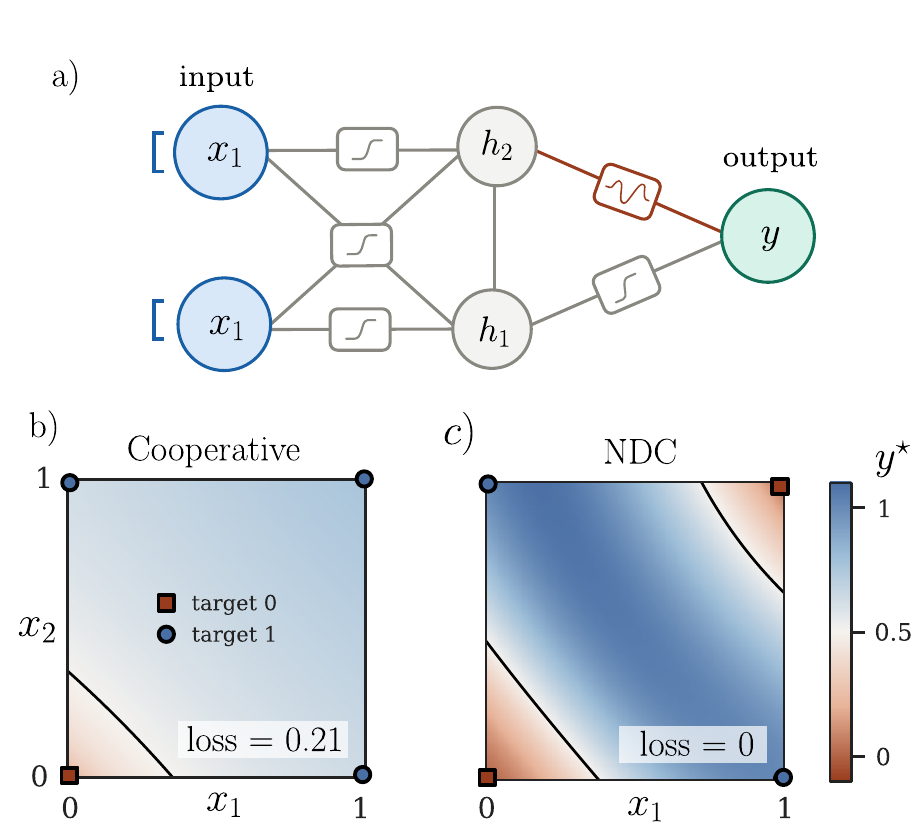}
\caption{\textbf{Quantum dot network implementing XOR.} Decision surfaces of a network of quantum dots trained on the XOR task, showing the steady-state output potential $y^{\star}$ as a function of the two input potentials $(x_1, x_2)$.~\textbf{(a)} The network structure used during training. In most runs of the numerical experiment the optimizer finds solutions with a single NDC edge~\textbf{(b)} When transport is restricted to the cooperative regime, the network cannot go beyond monotonic functions and converges to a solution with non-zero residual loss.~\textbf{(c)} With NDC enabled, the network can implement non-linear separations required by XOR. The corresponding decision boundary correctly partitions all four input points. Both networks share the same architecture and training protocol (gradient descent via the implicit value theorem). The two panels share the same architecture, optimizer, and initialization. The only difference is whether NDC is permitted.  
}
\label{fig:xor_network}
\end{figure}

In order to reach a lower error and realize the XOR function, the network must violate the assumptions of Theorem~\ref{thm:kamke}, namely the cooperative condition from Eq.~\eqref{eq:coop_condition}. In Fig.~\ref{fig:xor_network}b all quantum dots are allowed to take any value of parameters, and therefore can enter the NDC regime. As a result,  the network trained without constraints can establish clean separation between the two classes, as demonstrated by the black decision contour in Fig.~\ref{fig:xor_network}b. The interesting feature of this construction is that the network requires only a single NDC edge to correctly implement XOR, see Fig.~\ref{fig:xor_network}a for details.

\subsubsection{Scaling to image classification}
\label{sec:mnist}

The XOR example demonstrates the mechanism of folding due to NDC. We now test whether the same principle can be scaled up to realistic problems. For that we consider the MNIST handwritten digit classification benchmark~\cite{lecun2002gradient}: Classifying $28 \times 28$ grayscale images into ten digit categories ($0$--$9$).

Each input pixel is encoded as the chemical potential of an input reservoir, so that the total number of input reservoirs is $n_{\mathrm{in}} = 28^2 = 784$. The output layer consists of ten reservoirs, one per digit, with the predicted label determined by the reservoir with the highest steady-state potential. Between input and output, we place $n_h$ hidden reservoirs arranged in a single layer and connected by quantum dots with both the input and output reservoirs. 


The results of the classification are presented in Fig. \ref{fig:mnist}. There we show the performance of thermodynamic networks operating in different transport regimes as a function of the number of neurons in the hidden layer. Although the high dimensionality of the parameter space leads to challenging training scenarios where the network can converge to multistable solutions, the final networks achieve robust performance on this task.

\begin{figure}
    \centering
    \includegraphics[width=1\linewidth]{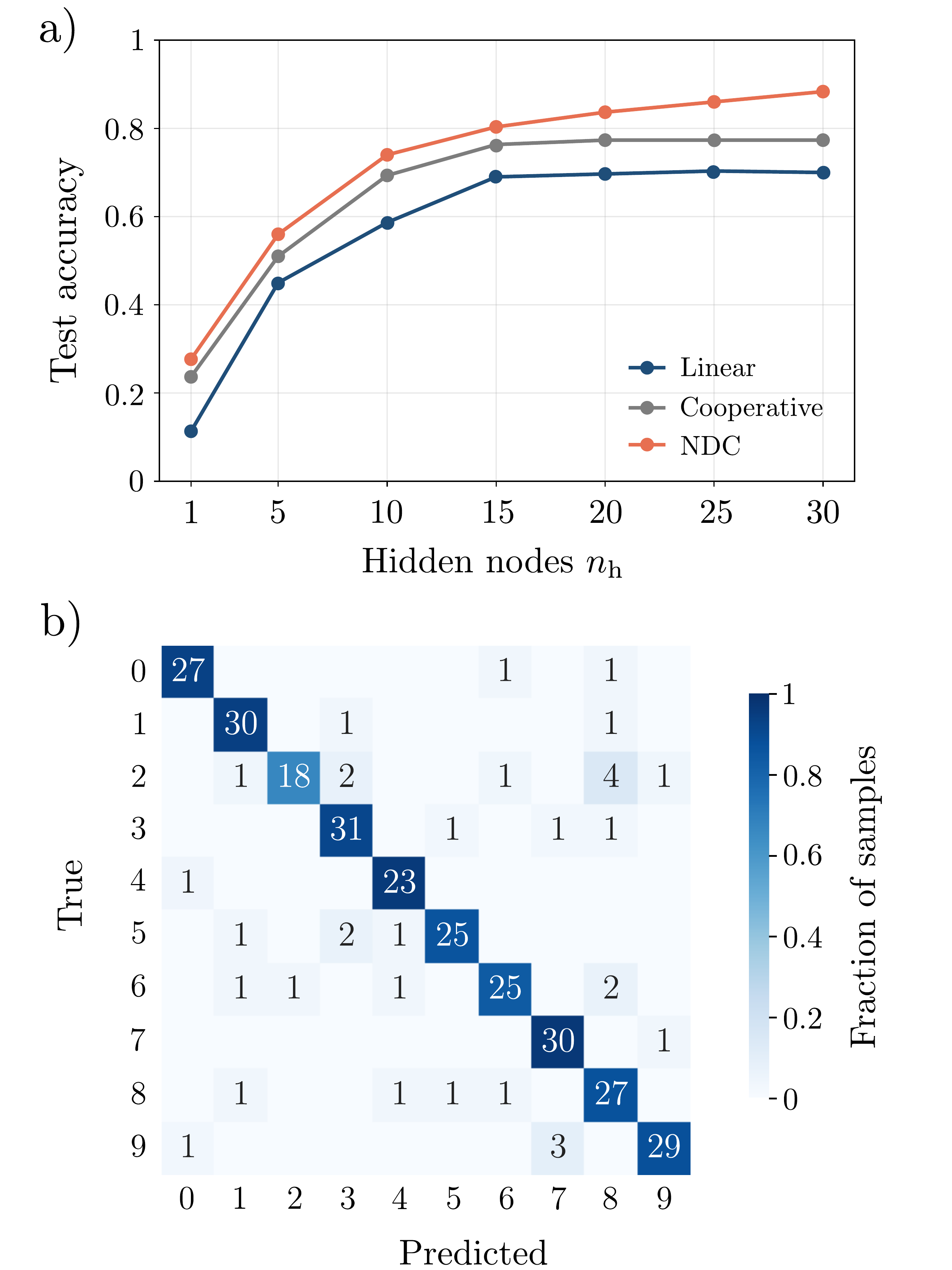}
    \caption{\textbf{Thermodynamic networks of quantum dots applied to MNIST problem (digit classification).} (a) Test accuracy of trained thermodynamic networks as a function of the number of neurons in the single hidden layer (784 input nodes, $n_{\rm h}$ hidden nodes, 10 output nodes). The curves correspond to three different regimes: linear, cooperative and NDC. All networks are trained on the MNIST dataset using gradient descent via implicit differentiation. We observe that NDC is a key ingredient for learning in this task, whereas linear and monotonic networks learn substantially worse. (b) Confusion matrix of the NDC-based thermodynamic network with $n_h = 30$ hidden nodes computed on the test set with $300$ test points and achieving $0.883$ accuracy. The diagonal shows which digits were recognized accurately, while off-diagonal entries capture which pairs of digits are hardest to distinguish.}
    \label{fig:mnist}
\end{figure}

\subsection{Enzymatic reaction networks}
\label{sec:chemical}
To illustrate the generality of thermodynamic networks, we now discuss a different physical scenario: enzymatic reaction networks (ERNs). In biology, such networks drive metabolism, energy production, and cellular
signalling~\cite{ghosh2024exploring}. Recent work has shown that engineered enzymatic systems can also perform computation, from single-layer perceptrons~\cite{pandi2019metabolic} to more complex nonlinear classification~\cite{okumura2022nonlinear}. In these demonstrations, networks of chemical reactions were explicitly designed to mimic feedforward neural network architectures using chemical components. The framework of thermodynamic networks can be applied to understand such systems, but from a fundamentally different perspective: we will show that \emph{any} reaction network exhibiting the appropriate transport physics can compute, regardless of whether it was designed to resemble a neural network or not.

In this setting, the conserved quantity is the molecular number, transport occurs via enzymatic catalysis and the NDC mechanism arises from the so-called \emph{substrate inhibition} phenomenon. The framework we developed in Sec.~\ref{sec:model} can therefore be directly applied in this case.

\subsubsection{Physical setup}
The nodes of the network are pools of molecular substrates, i.e. reservoirs that contain molecules of different chemical species with varying concentrations~$c_\alpha$. The transport between two nodes is mediated by an enzyme that catalyzes the conversion between compound~$\alpha$ into another compound~$\beta$. Each such reaction defines an edge of the network, with the enzyme playing the role of a nonlinear element coupling
the nodes and whose properties determine the
transport characteristics (see Fig.~\ref{fig:reaction_network}).

Each node $\alpha$ represents a distinct chemical compound characterized by its chemical potential,
\begin{align}
\label{eq:chem_potential_repeat}
\mu_\alpha = \mu_\alpha^\circ + k_B T \ln \frac{c_\alpha}{c_{\text{ref}}},
\end{align}
where $c_{\text{ref}}$ is the reference concentration and $\mu_{\alpha}^{\circ}$ is the potential at the reference concentration. Note that by inverting this relation, i.e. $c_\alpha = c_{\text{ref}}\,\exp[(\mu_\alpha - \mu_\alpha^\circ)/k_B T]$, any reaction flux that depends on concentrations can be expressed in terms of potentials, so that $I_{\alpha\beta}(c_\alpha(\mu_\alpha), c_\beta(\mu_\beta)) \equiv I_{\alpha\beta}(\mu_\alpha, \mu_\beta)$.

The node dynamics follows from conservation of molecular number. In a pool with a fixed volume $V$, the concentration of compound $\alpha$ changes according to the net influx from all reactions in which it participates in, that is $V\,\dot{c}_\alpha = - I_{\alpha}$. Rewriting this in terms of $\mu_\alpha$ and using the chain rule yields the same form as Eq.~\eqref{eq:dynamics} but now with a $\mu$-dependent capacity, namely
\begin{align}
\label{eq:chem_capacity_repeat}
C_\alpha(\mu_\alpha)\,\dot\mu_\alpha = -I_{\alpha}(\mu),
\end{align}
where $C_\alpha(\mu_\alpha) = V\,c_\alpha(\mu_\alpha)/k_B T$ and $I_{\alpha}$ is the net reaction flux into reservoir $\alpha$ from all neighboring chemical pools (see Eq.\eqref{eq:node_current}), similar to the current arising in networks of quantum dots discussed in Sec.~\ref{sec:quantum_dots}.

The computational \emph{inputs} are encoded in the chemical potentials of a subset of chemical pools that are held fixed. The remaining pools relax to a non-equilibrium steady state, with designated output pools encode the output. The \emph{parameters} of the network are the kinetic properties of each enzyme (e.g. conversion rates). These are set by molecular structure of the enzyme and can, in principle, be modified through techniques such as directed evolution or rational enzyme design~\cite{bornscheuer2012engineering}.

\subsubsection{Enzymatic transport}
A simple and well-studied model of reversible enzyme kinetics is the Michaelis--Menten mechanism~\cite{cornish2013fundamentals}, which expresses the reaction flux from $\alpha$
to $\beta$ as
\begin{align}
\label{eq:michaelis_menten}
I_{\alpha\beta}(c_{\alpha, c_{\beta}})
  = \frac{1}{{1 + \dfrac{c_\alpha}{K_\alpha}
            + \dfrac{c_\beta}{K_\beta}}} \left(\dfrac{V_f}{K_\alpha}\, c_\alpha
        - \dfrac{V_r}{K_\beta}\, c_\beta\right),
\end{align}
where $V_{f/r}$ are the tunable parameters that characterize the maximal forward and reverse rates, and $K_{\alpha/\beta}$, are the Michaelis constants that characterize compound binding. These parameters satisfy the analog of local detailed balance, known as Haldane relation~\cite{cornish2013fundamentals},
$V_f K_\beta / (V_r K_\alpha) = K_{\mathrm{eq}}$, where
$K_{\mathrm{eq}}$ is the equilibrium constant of the
reaction. This ensures that the flux vanishes at
equilibrium and entropy production is non-negative.

The flux in Eq.~\eqref{eq:michaelis_menten} increases
monotonically with substrate concentration $c_\alpha$ and
decreases monotonically with product concentration
$c_\beta$. A network built entirely from such reactions is
then cooperative in the sense of
Sec.~\ref{sec:cooperativity}. That is, all differential conductances are positive and the network computes only monotone functions.

\begin{figure}
\centering
\includegraphics[width=\linewidth]{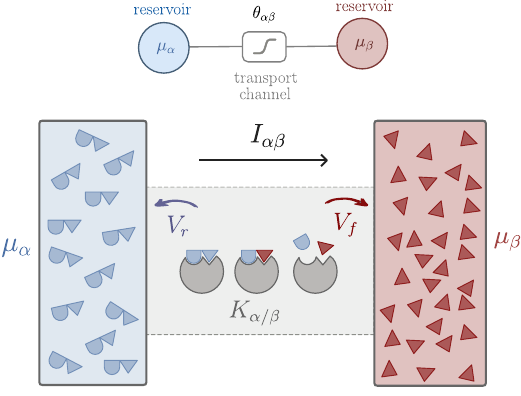}
\caption{\textbf{Enzymatic realization of a thermodynamic network edge.} The substrate and product pools act as mesoscopic reservoirs with chemical potentials $\mu_{\alpha}$ and $\mu_{\beta}$ respectively. An enzyme mediates transport between them by binding substrate molecules at its active site (with Michaelis constant $K_{\alpha}$), catalyzing their interconversion (with forward $V_f$ and reverse rate $V_r$), and releasing it into the product pool (with Michaelis constant $K_{\beta}$). The resulting Michaelis--Menten flux [Eq.~\eqref{eq:michaelis_menten}] specifies the edge current $I_{\alpha \beta}$. At high substrate concentrations, the active site of the enzyme is blocked, causing the flux to rise and then fall with concentration. This is the analogue of negative differential conductance that breaks cooperativity in thermodynamic networks.}
\label{fig:reaction_network}
\end{figure}

\subsubsection{NDC mechanism (substrate inhibition)}
To access non-monotone computation, we need a chemical analogue of NDC, namely a mechanism that makes the reaction flux respond non-monotonically to substrate concentration. This is a generic phenomenon made possible by \emph{substrate inhibition}~\cite{reed2010biological}. More specifically, at high concentrations, excess substrate molecules bind to
secondary sites on the enzyme, forming unproductive complexes that jam catalysis. As a consequence, the molecular flux rises, peaks, and then falls as substrate concentration increases. Substrate inhibition is common in biology, and affects many important metabolic enzymes~\cite{reed2010biological}.

The simplest model incorporating substrate inhibition adds an additional term to the denominator of Eq.~\eqref{eq:michaelis_menten}, see e.g. Refs. ~\cite{cornish2013fundamentals,murray2003mathematical,copeland2023enzymes}. The resulting expression for reaction flux becomes  
\begin{align}
\label{eq:substrate_inhibition}
I_{\alpha\beta}&(c_{\alpha}, c_{\beta}) 
  = \nonumber \\ 
  &\frac{1}{{1 + \dfrac{c_\alpha}{K_\alpha}
            + \dfrac{c_\beta}{K_\beta}} + \dfrac{c_\alpha^2}{K_\alpha K_i}} \left(\dfrac{V_f}{K_\alpha}\, c_\alpha
        - \dfrac{V_r}{K_\beta}\, c_\beta\right),
\end{align}
where $K_i$ is the inhibition constant characterizing the secondary binding site. The presence of the additional term $c_\alpha^2/(K_\alpha K_i)$ can be understood from the following heuristic argument. An enzyme molecule typically has a single ''slot" that can accommodate a single molecule of the substrate. After binding the enzyme, the substrate molecule is converted into the product. At high substrate concentrations, a second molecule can bind to this slot before the first one is processed. As a consequence, the reaction pathway is blocked until both molecules dissociate. Since this even requires the two substrate molecules to meet at the same enzyme, the jamming rate is proportional to $c_{\alpha}^{2}$, as opposed to the rate of the primary reaction which scales linearly with $c_\alpha$. At low substrate
concentrations ($c_\alpha \ll \sqrt{K_\alpha K_i}$), the inhibition term is negligible and the kinetics reduces to the standard Michaelis--Menten form. At high concentrations, the $c_\alpha^2$ term dominates the denominator and the flux decreases as $\sim K_i/c_\alpha$.

Let us now look at the differential conductances that enter the Jacobian. Differentiating the net flux from Eq.~\eqref{eq:substrate_inhibition} we obtain 
\begin{align}
\label{eq:chem_chainrule}
G_{\alpha\beta} = -\frac{c_\beta}{k_B T}\,\frac{\partial I_{\alpha}}{\partial c_\beta}.
\end{align}
Note that, because $c_{\ell}/(k_B T)>0$, the sign of each differential conductance equals the sign of the corresponding current derivative. 

Differentiating Eq.~\eqref{eq:substrate_inhibition} with respect to substrate concentration gives (see
Appendix~\ref{app:chem_conductances} for details),
\begin{align}
\nonumber
G_{\alpha\beta}
  &= \frac{1}{K_\alpha D_{\alpha\beta}^2}
    \Bigg[
      V_f\!\left(1 + \frac{c_\beta}{K_\beta}\right)
    + \frac{V_r\, c_\beta}{K_\beta}
      \!\left(1 + \frac{2c_\alpha}{K_i}\right) \\
\label{eq:chem_source_cond}
      &\quad\; - V_f\,\frac{c_\alpha^2}{K_\alpha K_i}
    \Bigg],
\end{align}
where $D_{\alpha\beta} = 1 + c_\alpha/K_\alpha
+ c_\beta/K_\beta + c_\alpha^2/(K_\alpha K_i)$. NDC in this model occurs when the substrate concentration $c_{\alpha}$ reaches a certain threshold value $c_\alpha^\star$.  In Appendix \ref{app:chem_conductances} we show that in the limit of small product concentration ($c_{\beta} \rightarrow 0)$, this threshold is given by 
\begin{align}
    c_\alpha^\star \approx \sqrt{K_\alpha K_i}.
\end{align}
Translating this into potentials via Eq.~\eqref{eq:chem_potential_repeat}, this gives a threshold potential
\begin{align}
    \mu_\alpha^\star \approx \mu_\alpha^\circ + k_B T \ln\!\frac{\sqrt{K_\alpha K_i}}{c_{\mathrm{ref}}},
\end{align}
which means that substrate inhibition (NDC) is reached whenever $\mu_\alpha>\mu_\alpha^\star$. In the case of arbitrary product concentration $\mu_{\beta}$, determining $\mu_{\alpha}^{\star}$ is also possible, see Appendix~\ref{app:chem_conductances} for details.

\subsubsection{Function regression}

\begin{figure}[t!]
\centering
\includegraphics[width=\linewidth]{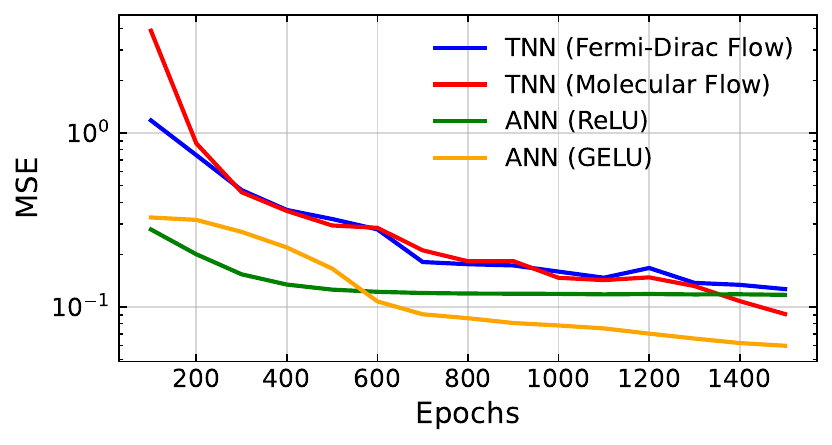}
\caption{\textbf{Performance of the different models in the function regression task.} Test loss as a function of the number of epochs for different models. All networks are trained (tested) on a 48 (12) points dataset corresponding to a function $y(x) = \exp(-x^2/10)\sin{(5x)}$ using gradient descent via implicit differentiation (TN) and backpropagation (ANN). Thermodynamic networks based on quantum dots (blue) and enzymatic reactions (red) are compared against ANNs with ReLU (green) and GELU (orange) activation functions. All models employ a 1-32-2-1 architecture. }
\label{fig:sin-square task}
\end{figure}
The numerical experiments performed on the XOR and MNIST tasks demonstrated classification with quantum networks. We now turn to regression, i.e. fitting a continuous mapping rather partitioning the input space using both presented platforms. This is a natural benchmark for thermodynamic networks, as the steady-state output potential is by default a continuous variable.  

We consider the task of approximating the function $y(x) = e^{-x^2/10} \sin(5x)$ which is non-monotone and thus requires network operating outside of the cooperative regime. The input coordinate $x$ is encoded in the potential of a single input reservoir ($n_{\rm in} = 1$) and the prediction is read as the steady-state potential of a single output reservoir. The network is arranged in a feed-forward architecture, with $2$ hidden layers, containing respectively $32$ and $2$ reservoirs.

In Fig.~\ref{fig:sin-square task} we compare the test loss of both discussed thermodynamic networks (quantum dots and enzymatic reactions) against standard ANNs with two types of activation functions (ReLU and GELU). All four models converge to similar test losses, showing that computation via thermodynamic networks is competitive with conventional neural network architectures on this task.

\section{Conclusions}

We developed a framework of thermodynamic networks, where the solution of a computation is retrieved from the NESS of the system. Notably, we identified NDC as a key physical property leading to strong expressibility and universal computation. In the absence of NDC, the network can implement non-linear but only monotonic functions, whereas full expressibility is achieved in the presence of NDC. Hence we establish a rigorous link betweeen physical transport properties and computational power.

To illustrate the relevance of our framework, we presented two possible implementations of thermodynamic networks, concretely quantum dots and enzymatic reactions. For both cases, we discussed the origin of NDC and demonstrated their potential for computing via standard benchmark tasks. In the future, it would be interesting to identify further instances of physical/chemical platforms where thermodynamic networks can be realized, and to identify specific problems for which thermodynamic networks are naturally tailored. Likewise, it would be desirable to establish links with other frameworks for physics-based computation \cite{opala2019neuromorphic,luppi2026competitive,dou2026training,li2025photonics,markovic2020physics,markovic2020quantum,Suderman2026,momeni2025training,lopezpastor2023self,mcmahon2023physics,mujal2021opportunities,opala2024roomtemperatureexcitonpolaritonneural,martinez2021dynamical}.

A natural future direction is to analyze the thermodynamics of computation with NESS by using tools from the fields of stochastic and quantum thermodynamics~\cite{seifert2012stochastic,wolpert2019stochastic,Wolpert2024,Chattopadhyay2025,deOliveiraJunior2025,Campbell2026,goold2016role,vinjanampathy2016quantum}. In such fields, links between computational power and thermodynamics have been identified at a fundamental level~\cite{Manzano2024,Meier2024Autonomous,Xuereb2026}, for logic gates using minimal models of transistors~\cite{gu2019microreversibility,wolpert2020thermodynamics,Gao2021,Freitas_2021,helms2022stochastic,Kuang2022Modelling,Klinger2026}  and recently for thermodynamic computing~\cite{rolandi2026energytimeaccuracytradeoffsthermodynamiccomputing}.  In our case, it would be interesting to analyze whether the presence of NDC comes with a higher thermodynamic cost, thus yielding a minimal energetic cost of full expressibility. Furthermore, it would be interesting to understand the role of fluctuations in our framework, where we expect that progress can be achieved by establishing connections with thermodynamic uncertainty relations~\cite{Horowitz2019}.

In the future, it would be interesting to explore further strategies for the training of the network, which ideally do not require an auxiliary digital computer (computation in-situ).  While our approach still requires solving a linear system of equations, one could imagine a training process based entirely on the dynamics of the system, in the spirit of equilibrium propagation~\cite{scellier2017equilibrium,ernoult2020equilibrium,Massar2024,wanjura2025quantum}. 

\section{Acknowledgements}
We thank Daniela Markovic, Giulia Rubino and Sparrow Suderman for discussions. The authors acknowledge funding from the Swiss State Secretariat for Education, Research and Innovation (SERI) under contract number UeM019-3. J.L.-P. and M.P.-L. acknowledge funding 
from  Spanish MICIN through the ATRAE Program (Grant ATR2024-154621 funded by MICIU/AEI/10.13039/501100011033)  as well as the project PID2024-162153NB-I00 (funded by  MICIU/AEI /10.13039/501100011033 and by  FEDER, UE). G.B. acknowledges Grant PID2023-151975NB-I00 funded by MICIU/AEI/10.13039/501100011033 and ERDF/EU. P.L.-B. acknowledges funding from Polish National Agency for Academic Exchange (NAWA) through
grant BPN/PPO/2023/1/00018/U/00001. G.H. acknowledges funding from NCCR SwissMAP as well as the Fondation Sandoz-Monique de Meuron program for academic promotion.

\bibliographystyle{apsrev4-2}
\bibliography{refs}

\onecolumngrid
\appendix

\section*{Supplementary Information}
\section{Microscopic dynamics of edges}
\label{app:microscopic}

In the main text, the edge current $I_{\alpha\beta}(\mu_\alpha, \mu_\beta)$ is treated as an instantaneous function of the adjacent potentials. Here we discuss its underlying microscopic dynamics using the formalism of master equations.

Consider an edge $e_{\alpha\beta} = (\alpha, \beta)$ as a microscopic system with discrete internal states $i \in \mathcal{I}$. The system exchanges the conserved quantity with the reservoirs located at its endpoints, hopping stochastically between states with transition rates $\Gamma_{ij}^{(\ell)}$ with $\ell \in \{\alpha, \beta\}$. The probabilities $p_i$ of occupying discrete states of the edge evolve via the master equation, 
\begin{align}
\label{eq:master}
\dot{p}_i = \sum_j \sum_{\ell \in \{\alpha, \beta\}} \big[\Gamma_{ij}^{(\ell)}\, p_j - \Gamma_{ji}^{(\ell)}\, p_i\big].
\end{align}
Assigning a value $q_i$ of the conserved quantity to each internal state, the current from reservoir $\alpha$ into the edge is
\begin{align}
\label{eq:micro_current}
I_{\alpha\beta} = \sum_{i,j} p_j^\star\, \Gamma_{ij}^{(\alpha)}(\mu_\alpha)\, (q_i - q_j),
\end{align}
where $p^\star$ is the steady-state distribution satisfying $\dot{p}_i = 0$. Conservation on the edge then ensures $I_{\alpha\beta} = -I_{\beta\alpha}$. 

The steady-state distribution $p^{\star}$ depends on $\mu_{\alpha}$ and $\mu_{\beta}$ through the rates $\Gamma_{ij}^{(\ell)}$, thus making the current a function of both potentials, $I_{\alpha \beta} = I_{\alpha \beta}(\mu_{\alpha}, \mu_{\beta})$. This functional dependence is exact, but it assumes that the edge has had time to reach its internal steady state. This is a reasonable assumption when the timescale of the edge dynamics is much shorter than the timescale of the node dynamics. Let us denote by $\tau_e$ the edge relaxation time, as set by the spectral gap of the rate matrix in Eq.~\eqref{eq:master}, and by $\tau_n$ the timescale of node dynamics set by the generalized capacitances $C_{\alpha}$. Then, this idealized description is valid in the \emph{adiabatic regime} $\tau_e \ll \tau_n$. Such a separation of timescales is common in various platforms, e.g. in quantum dot arrays where tunneling occurs on picosecond timescales while node (reservoir) charging usually takes nanoseconds \cite{hanson2007spins}.

\section{Universal Approximation Theorem for Thermodynamic Networks}
\label{app:universality}
In this Appendix we prove that thermodynamic networks can approximate any continuous function to arbitrary precision, provided the transport channels satisfy certain assumptions (Theorem~\ref{thm:universal}). The proof is constructive, i.e. we build a three-layer network whose steady-state map in an appropriate limit reduces to a single-hidden-layer neural network. Universal approximation of thermodynamic networks then follows by invoking a classical density argument. 

\subsection{Setting and assumptions}
\label{app:UAT:setup}

\noindent\textbf{Architecture.} Let $d,N\in\mathbb{N}$. Consider the following network ansatz:
\begin{itemize}
 \item $d$ \emph{input} reservoirs clamped at potentials $x=(x_1,\dots,x_d) \in K$, where $K\subset\mathbb{R}^d$ is a compact set. 
 
 \item $N$ \emph{hidden} reservoirs with unit capacity $C_{j} = 1$, each connected to a reference reservoir at a tunable potential $b_j\in\mathbb{R}$ using a linear transport channel with conductance $g_j = g_h >0$ for all $j$. 
 \item One \emph{output} reservoir of unit capacity, with constant conductance $g_o>0$ to a reference potential $b_o\in\mathbb{R}$.
 \item Every input--hidden pair $(i,j)$ is connected by an edge carrying current $I_{ij}(x_i,\mu_j;\theta_{ij})$. 
 \item Every hidden--output pair $(j,o)$ is connected by an edge carying current $\varepsilon_j\,\widetilde I_j(\mu_j,\mu_o;\phi_j)$, where $\varepsilon_j>0$ is an edge-specific coupling strength. We write $\varepsilon:=\max_j\varepsilon_j$ for the maximal coupling.
\end{itemize}
{\ \\}
\noindent\textbf{Assumptions.} To prove universality of thermodynamic networks we require the following assumptions:
\begin{description}
 \item[(P1) Arbitrary negative differential conductance] The network supports edges with tunable differential conductance, i.e. there exist physical parameters $\theta$ such that the differential conductance 
       \begin{equation*}
        G_{\alpha\gamma} := -\frac{\partial I_\alpha}{\partial \mu_\gamma}\Bigg|_{\mu'} \qquad \text{with} \quad  I_{\alpha} := \sum_{\beta \in \mathcal{N}}I_{\alpha \beta}(\mu_{\alpha}, \mu_{\beta}; \theta_{\alpha \beta})
       \end{equation*}
can take any value in $\mathbb{R}$ at any operating point $\mu'$.

\item[(P2) Antisymmetric edges] The network supports antisymmetric edges, namely
\begin{equation}
    I_{\alpha \beta}(\mu_{\alpha},\mu_{\beta})=I_{\alpha\beta}(\mu_{\alpha} - \mu_{\beta}) \qquad \text{with} \quad I_{\alpha \beta}(-V) = - I_{\alpha \beta} (V) \quad \forall \quad V\in\mathbb{R}.
       \label{eq:P2-anti}
      \end{equation}
\end{description}
{\ \\}
\noindent\textbf{Steady-state equations.} Define $F_j^0$ as the total current flowing into hidden node $j \in \{1, \ldots, N\}$ when $\varepsilon = 0$,
\begin{equation}
 F_j^{0}(x,\mu)
 \;:=\;\sum_{i=1}^{d} I_{ij}(x_i,\mu_j)\;-\;g_h(\mu_j-b_j).
 \label{eq:F0-def}
\end{equation}
Current balance at each time-dependent node then leads to the set of equations
\begin{align}
\text{(hidden)}\qquad
 F_j^{0}(x,\mu^{\star})
   &\;=\;\varepsilon_j\,\widetilde I_j(\mu_j^{\star},\mu_o^{\star}),
   \qquad j=1,\dots,N,
   \label{eq:ss-hidden}\\[2pt]
\text{(output)}\qquad
 g_o\,(\mu_o^{\star}-b_o) &\;=\;\sum_{j=1}^{N}\varepsilon_j\,\widetilde I_j(\mu_j^{\star},\mu_o^{\star}),
   \label{eq:ss-output}
\end{align}
In the case when output is uncoupled, i.e. $\varepsilon = 0$, the system of equations from Eq.~\eqref{eq:ss-hidden} splits into $N$ equations. We further denote with $\mu^{0}(x)$ the solution to the equation 
\begin{align}
    F_j^{0}(x, \mu^{0}(x))=0, \qquad \text{for all} \quad j,
\end{align}
i.e. the hidden steady state at vanishing couplings $\varepsilon_j=0$. We now proceed to the main theorem of this Appendix. 

\subsection{Theorem and Proof}
\label{app:UAT:statement}

\begin{theorem}[Universal approximation]
\label{thm:universal}
Under assumptions \emph{(P1)} and \emph{(P2)} above, for every compact
$K\subset\mathbb{R}^d$, every continuous $f\colon K\to\mathbb{R}$, and
every $\delta>0$, there exist $N\in\mathbb{N}$, $\nu > 0$ and
physical parameters $\{\theta_{ij},\phi_j,b_j\}$ such that the
steady-state map $N_\theta(x)=(\mu_o^{\star}(x)-b_o)/\nu$  satisfies
\begin{equation*}
 \sup_{x\in K}\bigl|N_\theta(x)-f(x)\bigr|\;<\;\delta.
\end{equation*}
\end{theorem}

\noindent Before we present the proof of the theorem, we first mention a few short lemmas.   

\begin{lemma}
 \label{lem:reflection}
Under \emph{(P2)}, the output-edge current flow
$\Phi_j(z):=\widetilde I_j(z,b_o)=\widetilde I_j(z-b_o)$ satisfies
\begin{equation}
 \Phi_j(b_o+u)\;=\;-\,\Phi_j(b_o-u)
 \qquad\forall \quad u\in\mathbb{R}.
 \label{eq:reflection}
\end{equation}
\end{lemma}
 
\begin{proof}
By Eq.~\eqref{eq:P2-anti}, we know that
$\Phi_j(b_o+u)=\widetilde I_j(u)$ and $\Phi_j(b_o-u)=\widetilde I_j(-u)$. It then follows that $\widetilde I_j(-u)=-\widetilde I_j(u)$, which then implies Eq. \eqref{eq:reflection}.
\end{proof}

\begin{lemma}
\label{lem:L1}
Let $\beta_j:=\mu_j^{0}(0)$ denote the steady-state of the hidden node $j$ when $x_i = 0$ for all $i \in \{1, \ldots, d\}$ and $\varepsilon = 0$. Assume further that the hidden node is stable there, namely
\begin{equation*}
 \widetilde g_j
 :=g_h\;-\;\sum_{i=1}^{d}
      \frac{\partial I_{ij}}{\partial\mu_j}\bigg|_{(0,\beta_j)} >0.
\end{equation*}
Then the steady-state potential of the hidden node $\mu_j^{0}(x)$ is given by
\begin{equation}
 \mu_j^{0}(x)
 =\beta_j+\sum_{i=1}^{d} w_{ij}\,x_i+O(|x|^{2}),
 \label{eq:L1-expansion}
\end{equation}
where the weights $w_{ij}$ are given by
\begin{equation}
 w_{ij}=\frac{G_{ij}}{\widetilde g_j},
 \qquad \text{with} \quad
 G_{ij}
  :=\frac{\partial I_{ij}}{\partial x_i}\bigg|_{(0,\beta_j)}.
 \label{eq:L1-weights}
\end{equation}
Each weight $w_{ij}$ is furthermore tunable to an arbitrary real value by the choice of the parameters of edge $(i,j)$.
\end{lemma}

\begin{proof}
The point $(x,\mu)=(0,\beta)$ solves 
$F_j^{0}(0,\beta)=0$. This is equivalent to $\beta_j$ being the solution to the equation 
\begin{equation}
 \sum_{i=1}^{d} I_{ij}(0,\beta_j)\;=\;g_h(\beta_j-b_j).
 \label{eq:L1-beta-def}
\end{equation}
Expanding $F_j^{0}$ for small $|x|$ around the point $(0, \beta)$ yields,
\begin{align}
 F_j^{0}(x,\mu) &= F_j^{0}(0,\beta) + \sum_{i=1}^{d}\left(\frac{\partial F_j^0}{\partial x_i}\Bigg|_{(0, \beta_j)}\!\! x_i \right)+ \left(\frac{\partial F_j^0}{\partial \mu_j}\Bigg|_{(0,\beta_j)} (\mu_j - \beta_j)\right) + \mathcal{O}(|x|^2) \\
 &=   F_j^{0}(0,\beta)  +\sum_{i=1}^{d} G_{ij}x_i-\widetilde g_j(\mu_j-\beta_j)+\mathcal{O}(|x|^{2}),
\end{align}
Setting $F_j^{0}=0$ at $\mu=\mu_j^{0}(x)$ and using $\widetilde g_j>0$ to
solve for $\mu_j^{0}-\beta_j$ yields Eqs.~\eqref{eq:L1-expansion}--\eqref{eq:L1-weights}. 
\end{proof}

 
\begin{lemma}\label{lem:L2}
For sufficiently small $\varepsilon$ the steady state $\mu^{\star}$ of the network satisfies
\begin{align}
 \mu_j^{\star}(x;\{\varepsilon_j\})
   &\;=\;\mu_j^{0}(x)\;+\;O(\varepsilon),
   \label{eq:L2-mu}\\[3pt]
 \mu_o^{\star}(x;\{\varepsilon_j\})
   &\;=\;b_o\;+\;\frac{1}{g_o}\sum_{j=1}^{N}
          \varepsilon_j\,\Phi_j\!\bigl(\mu_j^{0}(x)\bigr)\;+\;O(\varepsilon^{2}),
   \label{eq:L2-output}
\end{align}
where we denoted $\Phi_j(z):=\widetilde I_j(z,b_o)$.
\end{lemma}
 
\begin{proof}
The steady state condition at the output reservoir reads
\begin{equation}
 g_o(\mu_o^{\star}-b_o)
 =\sum_{j=1}^{N}\varepsilon_j\,
      \widetilde I_j\!\bigl(\mu_j^{\star},\mu_o^{\star}\bigr).
 \label{eq:L2-output-balance}
\end{equation}
When we set all $\varepsilon_j=0$ we have that $\mu_j^{\star} = \mu_j^{0}$ and $\mu_o^{\star}=b_o$. Expanding around that point gives
\begin{equation*}
 \widetilde I_j\!\bigl(\mu_j^{\star},\mu_o^{\star}\bigr)
 \;=\;\Phi_j\!\bigl(\mu_j^{0}(x)\bigr)\;+\;O(\varepsilon).
\end{equation*}
Because the right-hand side of~\eqref{eq:L2-output-balance} is already prefactored by $\varepsilon_j$, this $O(\varepsilon)$ correction becomes $O(\varepsilon^{2})$ in the formula for
$\mu_o^{\star}$. Substituting and dividing by $g_o$
yields Eq.~\eqref{eq:L2-output}.
\end{proof}
We are now ready to proceed with the proof of Theorem~\ref{thm:universal}.

{\ } \\
\noindent \emph{Proof of Theorem \ref{thm:universal}.}  Let us fix the target function $f\in C(K)$ and the desired accuracy parameter $\delta>0$. The proof consists of three main steps.

{\ } \\ 
\noindent \textbf{Step 1.} Using the expansion $\mu_j^{0}(x) =\beta_j+\sum_{i=1}^{d} w_{ij} x_i + \mathcal{O}(|x|^{2})$ from Lemma~\ref{lem:L1} in Eq.~\eqref{eq:L2-output} yields 
\begin{equation}
 \mu_o^{\star}(x;\{\varepsilon_j\})
\;=\;b_o\;+\;\sum_{j=1}^{N}\eta_j\,\Phi_j\!\bigl(\beta_j+w_j\!\cdot\!x\bigr)
      \;+\;R(x;\varepsilon),
 \qquad R=O(\varepsilon^{2})+O(|x|^{2}),
 \label{eq:UAT:raw}
\end{equation}
where we denoted $\eta_j:=\varepsilon_j/g_o$ and used $R$ to denote the remainder from Taylor expansion. By property (P2), we may take tall output edges to have identical characteristic, i.e. $\tilde{I}_j = \tilde{I} =: \sigma$, so that $\Phi_j(z)= \sigma(z-b_o)$. Setting $\tilde\beta_j:=\beta_j-b_o$ our thermodynamic network realises
\begin{equation}
 \mu_o^{\star}(x;\{\varepsilon_j\})
=b_o+\sum_{j=1}^{N}\eta_j\,\sigma\bigl(w_j\cdot x+\tilde\beta_j\bigr) + R(x;\varepsilon).
 \label{eq:UAT:family-new}
\end{equation}
Note now that due to property P2 and Lemma~\ref{lem:reflection}, the function $\sigma$ can be chosen to be antisymmetric. 

{\ } \\ 
\noindent \textbf{Step 2.} Since the function $\sigma$ is antisymmetric, for any $a\in\mathbb{R}^d$,
$\theta\in\mathbb{R}$, and $c\in\mathbb{R}$ we can write
\begin{equation}
 c\,\sigma(a\!\cdot\!x+\theta)
 =|c|\,\sigma\!\bigl(\operatorname{sign}(c)\,a\!\cdot\!x+\operatorname{sign}(c)\,\theta\bigr).
 \label{eq:UAT:signflip}
\end{equation}
In other words, we can trade a negative coefficient in front of $\sigma$ with a sign flip of the argument. We now need to show that the sums of the form $\sum_{j=1}^{N}\eta_j\,\sigma(w_j\cdot x+\tilde\beta_j)$ that appear in the steady-state form in Eq.~\eqref{eq:UAT:family-new} can approximate any function $f$ arbitrarily well. This follows directly from Hornik's universal approximation theorem combined with the property from Eq.~\eqref{eq:UAT:signflip}.
 
Hornik's theorem (Theorem~2 of~\cite{hornik1991approximation}) is a classical result in the theory of artificial neural networks. In its basic form, it establishes that for any continuous, bounded and non-constant function $\sigma$, the linear combinations
\begin{equation}
  S(x) := \sum_{i=1}^{M} c_i\,\sigma(a_i\cdot x+\theta_i),
 \qquad \text{with freely tunable parameters} \quad c_i\in\mathbb{R}, \quad a_i\in\mathbb{R}^d, \quad \theta_i\in\mathbb{R},
 \label{eq:UAT:hornik-sum}
\end{equation}
are dense in $C(K)$, meaning that they can approximate any real function defined on a domain $K$. When applied to our target function $f$, Hornik's theorem ensures there exist an integer $M$ and coefficients $\{c_i,a_i,\theta_i\}_{i=1}^{M}$ for
which Eq.~\eqref{eq:UAT:hornik-sum} approximates $f$ up to error $\delta/2$ on $K$. By~\eqref{eq:UAT:signflip} we may further assume that when $\sigma$ is antisymmetric, $c_i$ can be taken to be non-negative, that is $c_i\geq 0$. 
 
Notice that in Step $1$ we established that the sum of Eq.~\eqref{eq:UAT:hornik-sum} can be realized as the steady-state of the thermodynamic network, up to the remainder $R(x, \varepsilon)$. More specifically, observe that when we set 
\begin{align}
    N = M, \quad w_j=a_j, \qquad \tilde\beta_j=\theta_j, \quad \eta_j= \nu c_j \quad \text{ (or equivalently } \varepsilon_j=\nu g_o c_j).
\end{align}
With this choice Eq.~\eqref{eq:UAT:family-new} now becomes
\begin{equation}
 \mu_o^{\star}(x) = b_o + \nu\sum_{j=1}^{M} c_j\,\sigma\bigl(a_j\cdot x+\theta_j\bigr) + R(x;\varepsilon),
 \label{eq:UAT:scaled}
\end{equation}
so that $N_\theta(x)=\bigl(\mu_o^{\star}(x)-b_o\bigr)/\nu$ reproduces the Hornik sum of Eq.~\eqref{eq:UAT:hornik-sum} up to the rescaled remainder $R(x;\varepsilon)/\nu$. Note that $\varepsilon = \nu g_o \max_j c_j$ is now controlled by the single parameter $\nu$. It remains to show that $R(x;\varepsilon)/\nu$ can be made arbitrarily small. 

{\ } \\ 
\noindent \textbf{Step 3.} It remains to bound the remainder $R(x;\varepsilon)$ and show that it vanishes. For convenience we split it as
\begin{align}
    R(x;\varepsilon) = R_x(x) + R_{\varepsilon}(x),
\end{align}
where $R_{\varepsilon}$ collects the correction due to a finite coupling $\varepsilon$ (Lemma \ref{lem:L2}), and $R_x$ the correction due to the linearization at the hidden layer (Lemma~\ref{lem:L1}). Throughout we assume that the output characteristic $\tilde{I}$ is bounded with Lipschitz constant $L$. 
{\ } \\

\noindent \emph{Bounding term $R_{\varepsilon}$.} Let us start from Eq.~\eqref{eq:ss-output} which after plugging $\varepsilon_j=\nu g_o c_j$ reads
\begin{equation}
 \mu_o^{\star}(x) = b_o + \nu\sum_{j=1}^{N} c_j\,\widetilde I_j\bigl(\mu_j^{\star},\mu_o^{\star}\bigr).
 \label{eq:UAT:exact-out}
\end{equation}
By Lemma~\ref{lem:L2} we know that switching on the coupling (i.e. taking $\varepsilon > 0$) shifts the steady-state by $O(\varepsilon)$, that is, there exist a constant $C$ such that 
\begin{align}
    |\mu_j^{\star} - \mu_j^{0}(x)| \leq C \nu \qquad \text{and} \qquad |\mu_o^{\star} - b_o| \leq C \nu \qquad \text{for all} \quad x \in K.
\end{align}
Therefore, replacing the true steady-state $(\mu_j^{\star}, \mu_o^{\star})$ by the reference one $(\mu_j^0(x), b_o)$ introduces error
\begin{align}
    |\tilde{I}(\mu_j^{\star}, \mu_o^{\star}) - \tilde{I}(\mu_j^0(x), b_0)| \leq L (|\mu_j^{\star} - \mu_j^0(x)| + |\mu_o^{\star} - b_o|) \leq 2 L C \nu, 
\end{align}
where we used Eq.~\eqref{eq:UAT:exact-out} and the Lipshitz property of $\tilde{I}$. 

Using our notation from before we have $\tilde{I}(\mu_j^0(x), b_o) = \sigma(\mu_j^0(x) - b_o)$. Inserting this into Eq.~\eqref{eq:UAT:exact-out}, we obtain 
\begin{align}
 \mu_o^{\star}(x)-b_o
 = \nu\sum_{j=1}^{N} c_j\,\sigma\bigl(\mu_j^{0}(x)-b_o\bigr) + R_\varepsilon(x),
 \qquad
\sup_{x \in K} |R_\varepsilon(x)| \le A \nu^2,
\end{align}
where $A := 2 L C \sum_{j=1}^N c_j$ is independent of $\nu$ and of $x\in K$.
{\ \\}

\noindent \emph{Bounding term $R_{x}$.} The second remainder arises due to the error from replacing $\sigma(\mu_j^0(x) - b_o)$ by $\sigma_{w_j \cdot x + \tilde{\beta}_j}$, since by Lemma \ref{lem:L1} we have  $\mu_j^{0}(x)-b_o = w_j\cdot x+\tilde\beta_j + O(|x|^2)$. This linearization is accurate only for small inputs, while in general the actual domain of the input $K$ is fixed. To resolve it, we rescale the inputs: we replace $x \rightarrow x/\Lambda$ and $w_{ij} \rightarrow \Lambda w_{ij}$ for some $\Lambda \geq 1$. This leaves the linear term $w_j\cdot x$ unchanged, while shrinking the input set to $K/\Lambda$. Bounding this correction therefore leads to
\begin{align}
    \sup_{x \in K} |R_x(x)| \leq \frac{B}{\Lambda} \nu,
\end{align}
where $B$ is some constant independent from $\nu$ and $\Lambda$. Note we can take $\Lambda$ as large as desired, thus making the correction term arbitrarily small. 
{\ \\ }

\noindent \emph{Combining both bounds.} Let us now collect both contributions, namely
\begin{align}
     \mu_o^{\star}(x)-b_o
 = \nu\sum_{j=1}^{N} c_j\,\sigma\bigl(w_j \cdot x  + \tilde{\beta}_j\bigr) + R_\varepsilon(x;\varepsilon),
 \qquad
\sup_{x \in K} |R_\varepsilon(x; \varepsilon)| \le A \nu^2 + \frac{B}{\Lambda} \nu, 
\end{align}
After dividing the above equation by $\nu$, the function $N_{\theta}(\mu_o^{\star}(x) - b_o)/\nu$ reproduces the Hornik sum $S(x)$ up to an error $|R(x;\theta)|\nu \leq A \nu + B/\Lambda$. Combining this with our estimate from Step 2 we obtain,
\begin{align}
    \sup_{x \in K} |N_{\theta}(x) - f(x)| \leq \sup_{x \in K} |S(x) - f(x)| + A \nu + \frac{B}{\Lambda} \leq \frac{\delta}{2} + A \nu + \frac{B}{\Lambda}.
\end{align}
Choosing $\nu < \delta / (4A)$ and then $\Lambda > 4B/\delta$ makes the right-hand smaller than $\delta$, which completes the proof of Theorem~\ref{thm:universal}.
\section{Training thermodynamic networks}
\label{app:training}

\subsection{Implicit Differentiation}
\label{app:implicit_differentiation}
The starting point is to differentiate the steady-state condition 
$F(\mu^\star; \theta) = 0$ with respect to a parameter $\theta_k$. 
Using the chain rule, this gives
\begin{align}
\label{eq:sensitivity}
J^\star \frac{\partial \mu^\star}{\partial \theta_k}
   = -\frac{\partial F}{\partial \theta_k}\bigg|_{\mu^\star},
\end{align}
where $J^\star$ is the Jacobian of the network dynamics evaluated at its fixed point. Equation~\eqref{eq:sensitivity} is a linear system 
that determines how the steady state shifts when the parameter $\theta_k$ is varied. 

The crucial observation is that we do not need to compute the full sensitivity matrix $\partial \mu^\star / \partial \theta_k$. We only need its projection 
onto the loss gradient at the output, since
\begin{align}
\label{eq:gradient_projection}
\frac{\partial L}{\partial \theta_k} 
  = \sum_{\alpha} \frac{\partial L}{\partial \mu_{\alpha}^{\star}} \frac{\partial \mu_{\alpha}^\star}{\partial \theta_k} = \left(\frac{\partial L}{\partial \mu^{\star}}\right)^{\top} \frac{\partial \mu^\star}{\partial \theta_k} = g^\top P \,\frac{\partial \mu^\star}{\partial \theta_k},
\end{align}
where $g :=
\left.
\nabla_z \mathcal{L}(z,y)
\right|_{z=N_\theta(x)} $  is the loss gradient 
at the output node and $P$ projects onto the output nodes. Combining 
Eqs.~\eqref{eq:sensitivity} and~\eqref{eq:gradient_projection}, we 
can rewrite the parameter gradient as
\begin{align}
\frac{\partial L}{\partial \theta_k} 
   = -\,g^\top P (J^\star)^{-1}\, 
       \frac{\partial F}{\partial \theta_k}\bigg|_{\mu^\star}.
\label{eq:lk_training}
\end{align}
Notice now that in Eq.~\eqref{eq:lk_training} only the rightmost 
factor $\partial F/\partial \theta_k$ depends on which parameter we 
differentiate. Everything to its left, namely the row vector 
$g^\top P (J^\star)^{-1}$, is the same for every $\theta_k$. We can 
therefore compute it \emph{once} by solving a single adjoint 
equation, that is
\begin{align}
\label{eq:adjoint_system_app}
(J^\star)^\top \lambda = P^\top g,
\end{align}
and then read off all parameter gradients through cheap contractions,
\begin{align}
\label{eq:adjoint_formula2}
\frac{\partial L}{\partial \theta_k} 
   = -\,\lambda^\top \frac{\partial F}{\partial \theta_k}\bigg|_{\mu^\star}.
\end{align}
One linear solve replaces the $m$ separate solves that finite  differences would require.

\subsection{Equilibrium Propagation}
\label{app:eq_prop}
The second strategy we discuss is known as equilibrium propagation (EP)~\cite{scellier_equilibrium_2017}. This is a particularly appealing approach for physical implementations, because it replaces the linear algebra of the adjoint method with a second physical relaxation. Where implicit differentiation solves for the adjoint vector $\lambda$ using an external linear solver, EP lets the network itself perform the equivalent computation through its own dynamics.

The protocol proceeds in two phases. In the \emph{free phase}, the network relaxes to its unperturbed steady state $\mu^0$ with the input nodes clamped to the training input $x$. In the \emph{nudged phase}, a weak external force is applied at the output nodes, gently pushing them toward lower loss. Concretely, the perturbed steady state $\mu^\alpha$ satisfies
\begin{align} \label{eq:ep_perturbed}
F(\mu^\alpha;\, \theta) = \alpha\, P^\top g^\alpha,
\end{align}
where $g^{\alpha}$ is now the loss gradient evaluated at the nudged output and $0 < \alpha \ll 1$ controls the nudge strength. Physically, the right-hand side represents a small external current injected at the output nodes, biasing them toward the target. The system then relaxes to a new steady state $\mu^{\alpha}$ in which these external currents are balanced by the internal transport in the network.

The gradient information is encoded in the difference between the two steady states. Define the rescaled shift $\delta\mu := (\mu^\alpha - \mu^0)/\alpha$. Expanding Eq.~\eqref{eq:ep_perturbed} to first order in $\alpha$ around the free-phase steady state and using $F(\mu^0;\theta)=0$ gives
\begin{align} \label{eq:ep_linearised}
J^\star \delta\mu = P^\top g, \qquad \text{with} \quad J^\star := \frac{\partial F}{\partial \mu}\bigg|_{\mu^\star}
\end{align}
which is precisely the linear system that implicit differentiation solves algebraically, with $\delta \mu$ playing the role of $\lambda$. Here the network solves it instead through physical relaxation. The EP gradient estimator then takes the form 
\begin{align} \label{eq:ep_estimator}
\left[\frac{\partial L}{\partial \theta_k}\right]_{\text{est}} = -(\delta\mu)^\top \frac{\partial F}{\partial \theta_k}\bigg|_{\mu^\star},
\end{align}
where $[\cdot]_{\text{est}}$ denotes the estimator of the gradient. 

To compare this with the exact gradient, let us recall that implicit differentiation of the steady-state condition (see Eq.~\eqref{eq:lk_training}) yields
\begin{align} \label{eq:exact_grad}
\frac{\partial L}{\partial \theta_k}
   = -\,g^\top P \,(J^\star)^{-1}\,\frac{\partial F}{\partial \theta_k}\bigg|_{\mu^\star}.
\end{align}
Substituting $\delta\mu = (J^\star)^{-1} P^\top g$ from Eq.~\eqref{eq:ep_linearised} into Eq.~\eqref{eq:ep_estimator}, the estimator can be written in the same form, namely
\begin{align} \label{eq:ep_estimator_explicit}
\left[\frac{\partial L}{\partial \theta_k}\right]_{\text{est}}
   = -\,g^\top P \,(J^\star)^{-\top}\,\frac{\partial F}{\partial \theta_k}\bigg|_{\mu^\star}.
\end{align}
The exact gradient and the EP estimator therefore differ only in how the Jacobian matrix enters. Eq.~\eqref{eq:exact_grad} involves $(J^\star)^{-1}$, whereas Eq.~\eqref{eq:ep_estimator_explicit} uses $(J^\star)^{-\top}$. The estimator is exact if and only if $(J^\star)^{-\top} = (J^\star)^{-1}$, that is, whenever the Jacobian $J^\star$ is symmetric.

When the network dynamics derives from a potential, that is when $F = -\nabla_\mu \Phi$, the Jacobian $J^\star = -\nabla^2_\mu \Phi$ is symmetric and negative definite at a stable steady state. In this case we have $(J^\star)^{-\top} = (J^\star)^{-1}$, and the EP estimator matches the exact gradient. The EP gradient estimator then becomes equal to the true gradient and the formula simplifies to the canonical result of Ref.~\cite{scellier_equilibrium_2017}, namely
\begin{align} \label{eq:ep_canonical}
\frac{\partial L}{\partial \theta_k} = \lim_{\alpha \to 0} \frac{1}{\alpha}\Big[\frac{\partial \Phi}{\partial \theta_k}\Big|_{\mu^\alpha} - \frac{\partial \Phi}{\partial \theta_k}\Big|_{\mu^0}\Big].
\end{align}
This requires evaluating the potential at two steady states, without explicitly computing the Jacobian or solving a linear system. This is the setting in which EP was originally formulated, and where it is most natural.

Thermodynamic networks with asymmetric transport laws typically have a non-symmetric Jacobian. Decomposing $(J^\star)^{-1}$ into symmetric and antisymmetric parts, $(J^\star)^{-1} = S + A$, one can realize that the EP estimator captures only the symmetric part $S$. The corresponding bias is proportional to $A$ and reflects the rotational (non-conservative) component of the vector field. When this component is small, the EP estimator provides a good approximation. However, for strongly non-conservative dynamics, one can has to make some sacrifices: either accept the bias and use an approximate training signal, employ (computationally expensive) correction schemes that estimate the antisymmetric contribution~\cite{laborieux2022holomorphic}, or fall back to implicit differentiation. 
\section{From Landauer--B\"uttiker transport to the quantum-dot current law}
\label{app_LB_to_quantumdot}

In this appendix we show how the current used in Eq.~\eqref{eq:landauer} of the main text follows from the Landauer--B\"uttiker description of coherent transport. We also show how, for a single-level quantum dot, a bias-dependent transmission arises microscopically when the dot level shifts under the electric field generated by the potential drop between the two leads. This mechanism provides the microscopic origin of the negative differential conductance discussed in Sec.~\ref{sec:quantum_dots}.

\subsection{Landauer--B\"uttiker transport and bias-dependent transmission}
In the Landauer--B\"uttiker description, transport through a mesoscopic conductor is viewed as coherent scattering between electronic reservoirs. For a two-terminal device, the steady-state current reads~\cite{Datta1997,Blanter2000}
\begin{align}
I_{\alpha\beta}
=
\int \frac{dE}{2\pi}\,
T(E,V)
\left[
f_\alpha(E)-f_\beta(E)
\right],
\end{align}
where $f_\ell(E)$ is the Fermi--Dirac distribution of reservoir $\ell$, and $T(E,V)$ is the transmission function of the conductor.

The transmission function gives the probability that an electron with energy $E$ traverses the device from one lead to the other. Microscopically, it reflects the internal spectrum of the conductor, its coupling to the leads, and the electrostatic potential profile inside the device. In many nanoscale conductors this potential profile is itself modified by the applied bias, so that the transmission becomes explicitly bias dependent, $T(E,V)$ with $V \propto \mu_\alpha-\mu_\beta$~\cite{Nitzan2003,Nitzan2024}. Such bias-dependent transmission functions arise in a variety of experimentally accessible systems, including resonant tunnelling structures, molecular junctions, semiconductor quantum dots, and superconducting hybrid devices.

This property plays a central role in nonlinear transport. When the transmission peak shifts or is deformed as the bias increases, its overlap with the transport window defined by the Fermi functions may decrease at large driving. As a result, the current can stop increasing and may even decrease, giving rise to negative differential conductance.

\subsection{Electrostatic shift in a single-level quantum dot}

We now specialize to a single quantum dot described by one electronic level. In the absence of an applied bias the dot Hamiltonian reads
\begin{align}
H_0 = \epsilon_0\, d^\dagger d ,
\end{align}
where $d^\dagger$ and $d$ are the creation and annihilation operators of the dot level and $\epsilon_0$ is its bare energy.
When a bias is applied between the two leads, an electric field develops across the device. 
This field couples to the electronic charge on the dot through the electrostatic interaction. 
For a single localized level, this can be described by adding to the Hamiltonian the term
\begin{align}
H_{\mathrm{el}} = -e\, \mathbf{E}\cdot\mathbf{r}\; d^\dagger d ,
\end{align}
where $\mathbf{E}$ is the electric field generated by the potential drop between the leads and $\mathbf{r}$ denotes the position of the localized electronic state. If the electric field is approximately uniform across the dot region, the potential shift experienced by the level is proportional to the applied bias $V$. 
In this case the dot energy acquires a linear electrostatic shift, 
\begin{align}
\epsilon(V) = \epsilon_0 + \lambda V ,
\end{align}
where the coefficient $\lambda$ depends on the geometry of the device and on how the potential drop is distributed across the structure.
Expressing the bias in terms of the lead chemical potentials $V \propto \mu_\alpha-\mu_\beta$, the dot Hamiltonian can be written as
\begin{align}
H_{\mathrm{dot}} = \epsilon(V)\, d^\dagger d,
\qquad
\epsilon(V)=a+b(\mu_\alpha-\mu_\beta),
\end{align}
which corresponds to Eq.~\eqref{eq:level_energy} of the main text. 
The parameter $b$ therefore quantifies how strongly the resonant level follows the electric field generated by the applied bias.

\subsection{Lorentzian transmission from Green's functions}

Starting from the Hamiltonian introduced above, the transport properties of the system can be obtained using the Green's-function formalism. In particular, the retarded Green's function of the dot determines the transmission function entering the Landauer--B\"uttiker current.
For a single noninteracting level coupled to two leads, the retarded Green's function reads
\begin{align}
G^r(E)
=
\frac{1}{E-\epsilon(V)+i\Gamma/2},
\end{align}
where $\Gamma=\Gamma_\alpha+\Gamma_\beta$ characterizes the coupling to the leads. The advanced Green's function is the complex conjugate, $G^a(E)=\bigl(G^r(E)\bigr)^*$.
Within this framework the transmission function can be written as~\cite{Meir1992,Datta1997,Blasi2026}
\begin{align}
T(E,V)=\Gamma_\alpha\, G^r(E)\Gamma_\beta G^a(E),
\end{align}
which yields explicitly
\begin{align}
T(E,V)
=
\frac{\Gamma_\alpha\Gamma_\beta}
{(E-\epsilon(V))^2+(\Gamma/2)^2}.
\end{align}

We therefore recover the familiar Lorentzian transmission, whose center is located at the dot energy $\epsilon(V)$. Because $\epsilon(V)$ depends on the applied bias through the electrostatic shift, the transmission becomes explicitly bias-dependent.

\subsection{Weak-coupling limit}

Substituting the Lorentzian transmission into the Landauer expression gives
\begin{align}
I_{\alpha\beta}
=
\int \frac{dE}{2\pi}
\frac{\Gamma_\alpha\Gamma_\beta}
{(E-\epsilon(V))^2+(\Gamma/2)^2}
\left[
f_\alpha(E)-f_\beta(E)
\right].
\end{align}
We now consider the weak-coupling regime $\Gamma \ll k_B T$, in which the resonance is much narrower than the thermal broadening of the Fermi distributions~\cite{Blasi_PRR2024}. In this limit the Fermi functions vary slowly across the Lorentzian peak and can therefore be evaluated at the resonance energy $E=\epsilon(V)$. This gives
\begin{align}
I_{\alpha\beta}
\approx
\left[
f_\alpha(\epsilon(V))-f_\beta(\epsilon(V))
\right]
\int \frac{dE}{2\pi}
\frac{\Gamma_\alpha\Gamma_\beta}
{(E-\epsilon(V))^2+(\Gamma/2)^2}.
\end{align}
Using
\begin{align}
\int \frac{dE}{2\pi}
\frac{1}{(E-\epsilon)^2+(\Gamma/2)^2}
=
\frac{1}{\Gamma},
\end{align}
we obtain
\begin{align}
I_{\alpha\beta}
=
\frac{\Gamma_\alpha\Gamma_\beta}{\Gamma}
\left[
f_\alpha(\epsilon(V))-f_\beta(\epsilon(V))
\right].
\end{align}
Defining
$\gamma_{\alpha\beta}=\frac{\Gamma_\alpha\Gamma_\beta}{\Gamma}$, we finally recover
\begin{align}
I_{\alpha\beta}
=
\gamma_{\alpha\beta}
\bigl[
f_\alpha(\epsilon_{\alpha\beta})-f_\beta(\epsilon_{\alpha\beta})
\bigr],
\end{align}
which corresponds to Eq.~\eqref{eq:landauer} of the main text.

\section{Derivations for quantum dot networks.}
\label{app:qd_conductances}
This appendix provides the details of derivations of formulas used throughout Sec.~\ref{sec:quantum_dots}.

\subsection{Setup and notation}
Let us recall that the current through a quantum dot $(\alpha, \beta)$ can be expressed as
\begin{align}
I_{\alpha\beta} = \gamma \big[ f_\alpha(\epsilon) - f_\beta(\epsilon) \big],
\end{align}
where $f_\ell(E) = [1 + e^{\beta_\ell(E - \mu_\ell)}]^{-1}$ is the Fermi--Dirac distribution of reservoir $\ell$ at inverse temperature $\beta_\ell = 1/k_B T_\ell$, and the level energy depends on the reservoir potentials through the electrostatic shift, namely
\begin{align}
\epsilon_{\alpha \beta} = a_{\alpha \beta} + b_{\alpha \beta}\,(\mu_\alpha - \mu_\beta).
\end{align}
We suppress edge indices for readability and write $\gamma$ for $\gamma_{\alpha\beta}$, $b$ for $b_{\alpha\beta}$, etc. The chemical potentials enter the current in two ways: \emph{directly}, through the Fermi functions $f_\ell$, and \emph{indirectly}, through the level energy $\epsilon$. The interplay between these two channels determines the sign of the differential conductances and hence the cooperativity of the network.

\subsection{Useful identities}
\noindent The derivative of the Fermi function with respect to energy is
\begin{align}
\frac{\partial f_\ell}{\partial \epsilon} = -\beta_\ell\, \tilde{f}_\ell,
\end{align}
where we have introduced the thermal broadening factor
\begin{align}
\tilde{f}_\ell \equiv f_\ell(\epsilon)\big[1 - f_\ell(\epsilon)\big] \geq 0.
\end{align}
This quantity is maximal when the level is resonant with the chemical potential ($\epsilon = \mu_\ell$, giving $\tilde{f}_\ell = 1/4$) and vanishes exponentially when the level is far from resonance. The derivative with respect to the chemical potential at fixed energy is
\begin{align}
\frac{\partial f_\ell}{\partial \mu_\ell}\bigg|_{\epsilon} = \beta_\ell\, \tilde{f}_\ell,
\end{align}
which is non-negative: raising the chemical potential increases the occupation at any given energy. Note further the useful relation $\partial f_\ell / \partial \mu_\ell |_{\epsilon} = -\partial f_\ell / \partial {\epsilon}$. The derivatives of the energy level with respect to both potentials are given by
\begin{align}
\frac{\partial \epsilon}{\partial \mu_\alpha} = b, \qquad \frac{\partial \epsilon}{\partial \mu_\beta} = -b.
\end{align}

\subsection{Computing differential conductances}
\noindent By the chain rule applied directly to $f_\beta(\epsilon(\mu_\alpha, \mu_\beta))$ we obtain
\begin{align}
\frac{\partial f_\alpha(\epsilon)}{\partial \mu_\beta} & = \frac{\partial f_\alpha}{\partial \epsilon}\,\frac{\partial \epsilon}{\partial \mu_\beta} = (-\beta_\alpha \tilde{f}_\alpha)(-b) = b\,\beta_\alpha \tilde{f}_\alpha, \\
\frac{\partial f_\beta(\epsilon)}{\partial \mu_\beta} &= \frac{\partial f_\beta}{\partial \mu_\beta}\bigg|_\epsilon + \frac{\partial f_\beta}{\partial \epsilon}\,\frac{\partial \epsilon}{\partial \mu_\beta} = \beta_\beta \tilde{f}_\beta + (-\beta_\beta \tilde{f}_\beta)(-b) = (1+b)\,\beta_\beta \tilde{f}_\beta.
\end{align}
The total derivative of the current with respect to $\mu_\beta$ can be computed as
\begin{align}
\frac{\partial I}{\partial \mu_\beta} = \gamma\Big[b\,\beta_\alpha \tilde{f}_\alpha - (1+b)\,\beta_\beta \tilde{f}_\beta\Big].
\end{align}
Using the sign convention $G_{\alpha \beta} \equiv -\partial I_{\alpha} / \partial \mu_\beta$ introduced in Sec.~\ref{sec:model}, we obtain
\begin{align}
\label{eq:app_drain2}
G_{\alpha \beta} = \gamma\Big[(1+b)\,\beta_\beta \tilde{f}_\beta - b\,\beta_\alpha \tilde{f}_\alpha\Big].
\end{align}

\section{Derivations for enzymatic reaction networks}
\label{app:chem_conductances}

In this appendix, we compute the differential conductances for enzymatic reaction networks subject to substrate inhibition, as described in Sec.~\ref{sec:chemical}. For simplicity we work with concentrations throughout, and will relate the results to the chemical-potential conductances at the end.

The net flux through an enzyme connecting substrate $\alpha$ to product $\beta$ with substrate inhibition is given by Eq.~\eqref{eq:substrate_inhibition}. For readability, we again suppress the edge indices and write
\begin{align}
\label{eq:app_flux}
I = \frac{N}{D},
\end{align}
where we introduced the short-hand notation,
\begin{align}
\label{eq:app_ND}
N = \frac{V_f\, c_\alpha}{K_\alpha} - \frac{V_r\, c_\beta}{K_\beta}, \qquad D = 1 + \frac{c_\alpha}{K_\alpha} + \frac{c_\beta}{K_\beta} + \frac{c_\alpha^2}{K_\alpha K_i}.
\end{align}
We also compute the partial derivatives that will be used repeatedly, namely
\begin{align}
\label{eq:app_partials}
\frac{\partial N}{\partial c_\alpha} = \frac{V_f}{K_\alpha}, \qquad
\frac{\partial N}{\partial c_\beta} = -\frac{V_r}{K_\beta}, \qquad
\frac{\partial D}{\partial c_\alpha} = \frac{1}{K_\alpha} + \frac{2\, c_\alpha}{K_\alpha K_i}, \qquad
\frac{\partial D}{\partial c_\beta} = \frac{1}{K_\beta}.
\end{align}
Finally, the concentration derivatives of the flux can be expressed as
\begin{align}
\label{eq:app_quotient}
\frac{\partial I}{\partial c_\ell} = \frac{1}{D^2}\left(\frac{\partial N}{\partial c_\ell}\, D - N\, \frac{\partial D}{\partial c_\ell}\right), \qquad \ell \in \{\alpha, \beta\}.
\end{align}

\subsection{Computing current derivatives}
Let us now calculate the source (substrate) current derivative $\partial I / \partial c_{\alpha}$. Writing Eq. Eq.~\eqref{eq:app_quotient} for $\ell = \alpha$ gives us
\begin{align} \label{eq:app_quotient_a}
\frac{\partial I}{\partial c_\alpha}
  = \frac{1}{D^2}\left[\frac{\partial N}{\partial c_\ell}D
    - N \frac{\partial D}{\partial c_\alpha} \right].
\end{align}
Expanding the above using Eqs.~\eqref{eq:app_ND}--\eqref{eq:app_partials} leads to
\begin{align}
\frac{\partial N}{\partial c_\ell}D
  &= \frac{V_f}{K_\alpha}
    + \frac{V_f\, c_\alpha}{K_\alpha^2}
    + \frac{V_f\, c_\beta}{K_\alpha K_\beta}
    + \frac{V_f\, c_\alpha^2}{K_\alpha^2 K_i}, \\
N\,\frac{\partial D}{\partial c_\alpha}
  &= \frac{V_f\, c_\alpha}{K_\alpha^2}
    + \frac{2\, V_f\, c_\alpha^2}{K_\alpha^2 K_i}
    - \frac{V_r\, c_\beta}{K_\alpha K_\beta}
    - \frac{2\, V_r\, c_\alpha\, c_\beta}{K_\alpha K_\beta K_i}.
\end{align}
After simplifying Eq.~\eqref{eq:app_quotient_a} becomes
\begin{align}
\label{eq:app_source_general}
\frac{\partial I}{\partial c_\alpha}
  = \frac{1}{K_\alpha\, D^2}\left[
    V_f\!\left(1 + \frac{c_\beta}{K_\beta}\right)
    + \frac{V_r\, c_\beta}{K_\beta}\!\left(1 + \frac{2\, c_\alpha}{K_i}\right)
    - \frac{V_f\, c_\alpha^2}{K_\alpha K_i}
  \right].
\end{align}
This expression has two positive contributions (the first two terms) and one negative contribution (the last term), which arises directly from substrate inhibition. The positive terms reflect the standard catalytic response: increasing substrate drives more flux. The negative term reflects the inhibition: at high concentrations, excess substrate jams the enzyme.

The source conductance changes sign when the expression in square brackets in Eq.~\eqref{eq:app_source_general} vanishes. This is a quadratic equation in $c_\alpha$, that is
\begin{align}
\frac{V_f}{K_\alpha K_i}\, c_\alpha^2
  - \frac{2\, V_r\, c_\beta}{K_\beta K_i}\, c_\alpha
  - V_f\!\left(1 + \frac{c_\beta}{K_\beta}\right)
  - \frac{V_r\, c_\beta}{K_\beta}
  = 0.
\end{align}
Since the constant term is strictly negative, this quadratic has exactly one positive root, giving the threshold
\begin{align}
\label{eq:app_threshold}
c_\alpha^\star
  = \frac{V_r\, K_\alpha\, c_\beta}{V_f\, K_\beta}
    + \sqrt{
      \frac{V_r^2\, K_\alpha^2\, c_\beta^2}{V_f^2\, K_\beta^2}
      + K_\alpha K_i\!\left(1 + \frac{(V_f + V_r)\, c_\beta}{V_f\, K_\beta}\right)
    }.
\end{align}
For $c_\alpha < c_\alpha^\star$ the substrate conductance is positive (cooperative), and for $c_\alpha > c_\alpha^\star$ it turns negative (NDC regime). In the limit of small product concentration ($c_\beta \to 0$), this reduces to the simple expression $c_\alpha^\star = \sqrt{K_\alpha K_i}$ given in the main text.

We now compute the current derivatives $\partial I / \partial c_{\beta}$. Substituting Eqs.~\eqref{eq:app_ND}--\eqref{eq:app_partials} into Eq.~\eqref{eq:app_quotient} for $\ell = \beta$ gives
\begin{align}
\frac{\partial I}{\partial c_\beta}
 = \frac{1}{D^2}\left[-\frac{V_r}{K_\beta}\, D
    - N \cdot \frac{1}{K_\beta}\right]
  = -\frac{1}{K_\beta\, D^2}\Big[V_r\, D + N\Big].
\end{align}
Expanding the term $V_r\, D + N$ leads to
\begin{align}
V_r\, D + N
  &= V_r\!\left(1 + \frac{c_\alpha}{K_\alpha} + \frac{c_\beta}{K_\beta} + \frac{c_\alpha^2}{K_\alpha K_i}\right)
    + \frac{V_f\, c_\alpha}{K_\alpha} - \frac{V_r\, c_\beta}{K_\beta} \notag \\
  &= V_r\!\left(1 + \frac{c_\alpha}{K_\alpha} + \frac{c_\alpha^2}{K_\alpha K_i}\right)
    + \frac{V_f\, c_\alpha}{K_\alpha}.
\end{align}
The flux derrivative can therefore be written as
\begin{align}
\label{eq:app_drain}
-\frac{\partial I}{\partial c_\beta}
  = \frac{1}{K_\beta\, D^2}\left[
    V_r\!\left(1 + \frac{c_\alpha}{K_\alpha} + \frac{c_\alpha^2}{K_\alpha K_i}\right)
    + \frac{V_f\, c_\alpha}{K_\alpha}
  \right].
\end{align}
Since all kinetic parameters ($V_f, V_r, K_\alpha, K_\beta, K_i$) and concentrations ($c_\alpha, c_\beta$) are strictly positive, every term in the square bracket is positive. Hence $-\partial I / \partial c_\beta > 0$ for all physically accessible concentrations, and the drain side is unconditionally cooperative.

Physically, this is because increasing product shifts the reaction backward, providing a restoring feedback that stabilizes the product node. Unlike the substrate side, no inhibition mechanism counteracts this stabilizing response.
\end{document}